\documentclass[conference]{IEEEtran}
\IEEEoverridecommandlockouts
\usepackage{tabularx}

\ifCLASSINFOpdf
\else
\fi
\usepackage[numbers]{natbib}
\usepackage{epsfig}
\usepackage{epstopdf}
\usepackage[T1]{fontenc}
\usepackage{amssymb}
\usepackage{amsmath}
\usepackage{amsfonts}
\usepackage{verbatim}
\usepackage{graphicx}
\usepackage{hyperref}
\usepackage[usenames,dvipsnames]{xcolor}
\usepackage{lipsum}
\usepackage{booktabs}
\usepackage[normalem]{ulem}

\usepackage{multicol}
\usepackage{pgfplots}
\usepackage{lipsum}
\usepackage{colortbl}
\usepackage{caption}
\usepackage{algorithm}
\usepackage[noend]{algpseudocode}
\usepackage{enumitem}
\makeatletter
\def\algbackskip{\hskip-\ALG@thistlm}
\makeatother
\usepackage{amsthm}

\newtheorem{theorem}{Theorem}

\usepackage{multirow}

\usepackage{colortbl}
\usepackage{tabularx}
\usepackage{lipsum}  
\usepackage{pifont}
\usepackage{makecell}
\usepackage{float}
\usepackage{amsthm}
\usepackage{pgf}
\usepackage[table]{xcolor}
\usepackage{threeparttable}
\newcommand\subsubsubsection{\@startsection{paragraph}{4}{\z@}%
  {1.5ex \@plus 1ex \@minus .2ex}%
  {-1em}%
  {\normalfont\normalsize\bfseries}}
\makeatother

\usepackage{tikz,bm,color}

\usetikzlibrary{circuits.logic.US,circuits.logic.IEC,fit}

\begin{document}
%
\title{A Lightweight Fault-Detection Scheme for Barrett Modular Multiplication Using Multiple Conditional Reduction Paths}

 \author{
   \IEEEauthorblockN{Rourab Paul\IEEEauthorrefmark{1}, \IEEEmembership{Member,~IEEE,} Paresh Baidya\IEEEauthorrefmark{2}, Krishnendu Guha\IEEEauthorrefmark{3}, Amlan Chakrabarti\IEEEauthorrefmark{4}}
      \IEEEauthorblockA{\IEEEauthorrefmark{1}Department of Computer Science and Engineering, Shiv Nadar University, Chennai, Tamil Nadu, India}
      \IEEEauthorblockA{\IEEEauthorrefmark{2}Department of Computer Science and Engineering, Siksha ‘O’ Anusandhan Deemed to be University, Bhubaneswar, India}
    \IEEEauthorblockA{\IEEEauthorrefmark{3}School of Computer Science and Information Technology, University College Cork, Ireland\\
    \IEEEauthorblockA{\IEEEauthorrefmark{4}School of IT, University of Calcutta, West Bengal, India}     
            \IEEEauthorblockA{\IEEEauthorrefmark{1}rourabpaul@snuchennai.edu.in}
    }
 }
\maketitle
\vspace{-10pt}

\begin{abstract}
Polynomial multiplication is the most resource-, time-, and energy-critical operation in lattice-based Post-Quantum Cryptography (PQC) and Fully Homomorphic Encryption (FHE) schemes. Lattice-based PQC schemes such as Kyber and Dilithium have already been standardized, while lattice-based FHE schemes such as BGV, BFV, and CKKS are widely recognized as leading candidate in FHE area. Barrett Modular Multiplication (BMM) for polynomial multiplication is widely adopted in PQC and FHE hardware accelerators due to its hardware friendly nature and efficient modular reduction capabilities. However, Side-Channel Attacks (SCAs) and Hardware Trojans may introduce intentional faults, while aging and various other factors can cause unintentional faults. These faults may target the $BMM$ unit, one of the most critical components of PQC and FHE infrastructures, potentially leading to information leakage and compromising system security. In this paper, we employ a Statistical Reduction Monitoring (SRM) method to protect the $BMM$ unit against such adversarial conditions. The proposed approach incurs minimal hardware overhead while providing efficient detection of both random and burst faults under both permanent and transient fault conditions.
\end{abstract}


\section{Introduction}
\label{sec:intro}
Decimation-in-Time (DIT) and Decimation-in-Frequency (DIF) Number Theoretic Transform (NTT) is the backbone of lattice-based post-quantum cryptography (PQC) and fully homomorphic encryption (FHE). It significantly accelerates polynomial multiplication by reducing the computational complexity of multiplying two degree-$(n-1)$ polynomials from $\mathcal{O}(n^2)$ to $\mathcal{O}(n\log n)$. Therefore, polynomial multiplication in NIST-standardized lattice-based cryptographic schemes such as Kyber \cite{FIPS203} and Dilithium \cite{FIPS204}, as well as in FHE schemes including Brakerski Fan Vercauteren (BFV) \cite{BFV} and Brakerski Gentry Vaikuntanathan (BGV) \cite{BGV}, which are part of the NIST FHE standardization framework, can be significantly accelerated through dedicated hardware accelerators based on the $NTT$. DIT uses the Cooley Tukey (CT) butterfly structure, whereas DIF employs the Gentleman Sande (GS) butterfly structure. Both the CT and GS butterfly units require modular multiplication as a fundamental operation. Barrett \cite{barrett} modular multiplication (BMM) is one of the most efficient and widely adopted algorithms for performing modular multiplication in lattice-based cryptographic schemes. However, due to the globalization of the silicon supply chain, hardware Trojans have emerged as a significant threat to PQC and FHE hardware accelerators. Since the $BMM$ is one of the most resource-, power-, and time-critical components in PQC and FHE implementations, it represents an attractive target for adversaries. By compromising the $BMM$, attackers can disrupt the functionality of the entire PQC and FHE infrastructure deployed in secure processors. As a result, it potentially leading to erroneous computations, denial-of-service attacks, or leakage of sensitive information.
\subsection{Threat Model}
\label{sec:threat}
Lattice-based PQC schemes often require a sampler to generate errors and secrets. Thomas Espitau et al. \cite{espitau} perform a loop-abort attack on the samplers used in CRYSTALS-Kyber, FrodoKEM, BLISS, and the GPV signature scheme.
 The error or secret polynomial in the lattice-based PQC schemes mentioned above can be written as
$e = [e_0, e_1, e_2, \ldots, e_{n-1}]$.
Espitau et al. \cite{espitau} show that by aborting the sampling process before completion, an attacker can obtain a lower-degree polynomial of the form
$e = [e_0, e_1, \ldots, e_{p-1}, 0, 0, \ldots],$
where the remaining coefficients are zero due to the premature termination of the sampler. As a result, this breaks the intended masking and leaks linear equations that reveal information about the secret. The $NTT$ operations of Kyber and Dilithium in the PQC library for ARM Cortex-M4 (PQM4) \cite{pqm4} were successfully attacked by \cite{ravi_ntt} Ravi et al. They targets the modular multiplication operation of $NTT$ and makes the twidle factor (root of unity) zero which produces a lower degree polynomial at the output of $NTT$. It is to be noted that, in lattice-based PQC and FHE schemes, polynomial multiplication is typically performed between a twiddle factor and another operand, which is a coefficient of the transformed polynomial and may originate from secret vectors, error polynomials, ciphertext polynomials, or other intermediate polynomials. In Kyber, the key generation, encryption, and decryption processes pass the secret polynomial $s$, error polynomial $e$, ephemeral secret $r$, noise polynomials $e_1$ (vector) and $e_2$ (scalar), and the first ciphertext component $u$ to the $NTT$ operation \cite{FIPS203}. In CKKS, the key generation, encryption, and decryption processes involve the secret polynomial $s$, public polynomial $a$, error polynomial $e$, ephemeral secret $v$, and noise polynomials $e_0$ and $e_1$, with ciphertext components $(c_0, c_1)$ \cite{ckks}. Similarly, other lattice-based PQC schemes such as Dilithium, NTRU and Falcon, as well as FHE schemes including BGV and BFV, require multiple polynomials to be transformed into the $NTT$ domain for efficient polynomial multiplication. Therefore, zeroizing attacks \cite{ravi_ntt} and loop-abort attacks \cite{espitau} can produce lower-degree polynomials at the $NTT$ output, thereby reducing the entropy of the secret and error polynomials in PQC and FHE schemes.
\subsection{Literature}
The fault-detection literature for PQC and FHE can be broadly divided into two categories in this context.
\subsubsection{Fault Detection for PQC and FHE} Several fault detection architectures have been proposed as error detection mechanisms for various cryptographic primitives, which are the fundamental building blocks of PQC and FHE protocols.
Fault detection techniques, including cyclic redundancy check (CRC) \cite{canto}, RENO \cite{sarker}, RESO \cite{canto2}, REMO \cite{saeed} and RECO \cite{canto2} are commonly adopted due to their low implementation overhead and ease of integration.
In paper \cite{canto2}, Canto et al. propose recomputation-based fault detection  architectures for lattice-based key encapsulation mechanisms (KEMs) FrodoKEM, Saber, and NTRU on a Kintex Ultrascale+ FPGA. These schemes use multiplication units such as matrix-by-matrix, matrix-by-vector, vector-by-vector, and polynomial multiplications. Their work introduce Re-computing with Shifted Operands (RESO), Re-computing with Negated Operands (RENO) and Re-computing with Scaled Operands (RECO) for the Multiply-Accumulate (MAC) operation to detect both transient and permanent faults with high error coverage and minimal performance overhead. 
In artcile\cite{canto}, Canto et al. presented a CRC-based error detection technique for finite-field multipliers in the Luov post-quantum signature scheme. The proposed CRC-5-based architecture is implemented on a Kintex UltraScale+ FPGA and achieves high fault coverage with acceptable hardware overhead. In the article \cite{sarker}, Sarker et al. have implemented fault detection architectures for hardware/software co-design implementations of the $NTT$ accelerator on Zynq UltraScale+ and Spartan-7 devices. They propose error detection through recomputing with negated operands (RENO) and recomputing with negated and swapped operands to detect both transient and permanent faults in the $NTT$ accelerator. The result demonstrates that placing RENO deeper in the logic path increases slice utilization but increases error detection efficiency exceeding 99\% against stuck-at fault models.
\par
Saeed et al. \cite{saeed} proposed a lightweight fault detection architecture for modular exponentiation $C=X^Y \bmod q$, a crucial operation of numerous
cryptographic applications such as diffie-Hellman key exchange, RSA cryptosystems, ElGamal encryption, and some PQC schemes. The proposed Recomputation with Modular Offset (REMO) computes $C'=(X+Offset)^Y~\bmod~q$. Then it compares the result of  the modular exponentiation unit, $C$, and the REMO unit, $C'$, to detect faults. This method achieves nearly 100\% error detection with minimal computational and area overhead.

\subsubsection{BMM Specific fault Detection}
To protect the $BMM$, serves as a critical modular multiplication unit within the $NTT$ of lattice-based PQC and FHE schemes, many countermeasures and security mechanisms have been proposed in the literature. Aghapour et al. \cite{aghapour_barrett} implements safeguards for $BMM$ and Barrett Modular Reduction ($BMR$). For $BMR$ they used a probability approximation. As shown in Eq.~\ref{eq:bmr1}, two intermediate remainders, $r_1$ and $r_2$, are computed to determine the final remainder $r$. Here, $N$ is a $p$-bit integer, while $u$ is a $2p$-bit integer. $\hat{q}$ is precomputed constant.
\begin{equation}
    r_1= u \bmod b^{p+1}, r_2= (\hat{q}N) \bmod b^{p+1},
    r=r_1-r_2
\label{eq:bmr1}
\end{equation}
The final remainder is calculated by a while loop as shown in Eq. \ref{eq:bmr2}. The theoretical property of $BMM$ says that there is approximately a $99\%$ probability that only a single subtraction will be required within the final correction loop, whereas there is approximately a $1\%$ probability that two subtraction operations will be required before obtaining the final remainder.
\begin{equation}
\textbf{while } (r \geq N) \textbf{ do } r \leftarrow r - N
\label{eq:bmr2}
\end{equation}
For $BMM$, Aghapour et al. \cite{aghapour_barrett} uses Recomputation with Modular Offset (REMO) scheme. To detect fault of $ab$ $mod$ $q$, they calculate $(a+k_1q)(b+k_2q)$ $mod$ $q$ where $k_1$ and $k_2$ are random integers. In a fault-free situation, both computations produce same outputs. However, any mismatch between the two outputs is interpreted as an a fault, and an error flag is asserted.
\par Baidya et al.~\cite{baidya} proposed three recomputation-based fault detection techniques, namely Recomputation with Negated Operand (RENO), Recomputation with Shifted Operand (RESO), and Recomputation with Swapped Operand (RESWO). Baidya et al.~\cite{baidya} proposed protection schemes for the $BMM$ and reported the fault-detection overhead at the $NTT$ level. However, the overhead of RESO, RECO, and RESWO recomputation-based fault detection schemes was not evaluated separately for the $BMM$. If quantified at the $BMM$ level, the additional computations introduced by recomputation would likely result in a significant overhead.
\subsection{Contribution}
The existing fault detection schemes for BMMs are predominantly recomputation-based, resulting in significant area and energy overheads. In this work, we propose a word-wise BMM architecture consisting of two reduction stages. By exploiting the probabilistic behavior of the conditional execution paths within these reduction stages, we design a Statistical Reduction Monitor (SRM) for fault detection. The main contributions of this work are summarized as follows: 
\begin{itemize}
    \item \textbf{Statistical Reduction-Aware Fault Detection:} The proposed fault detection scheme exploits the statistical behavior of the two reduction stages employed in the word-wise BMM. Unlike existing recomputation-based approaches, the proposed Statistical Reduction Monitor (SRM) leverages the intrinsic probability distribution of reduction-path execution to detect faults. 
    \item \textbf{Ultra-Lightweight Protection for BMM Architectures:} The proposed scheme achieves efficient fault detection with minimal resource and energy overhead, making it one of the lightest-weight protection mechanisms specifically designed for BMM architectures. The protected BMM is also integrated into the $NTT$, and to the best of our knowledge, it achieves the lowest area, delay, and energy overhead among all existing protected $NTT$ schemes.
    \item \textbf{Comprehensive Validation Against Diverse Fault Models:} The extensive simulation and hardware-emulation results demonstrate that the proposed word-wise $BMM$ achieves high fault-detection efficiency against both random and burst fault injections during the Kyber and CKKS key-generation processes. The proposed fault-detection scheme effectively detects both transient and permanent faults while incurring minimal hardware overhead.

\end{itemize}
This paper is organized as follows. Sec. \ref{sec:intro} introduces the problem statement. Sec. \ref{sec:BMM} presents the algorithm of the proposed word-wise BMM. Sec. \ref{sec:srm} describes the proposed Statistical Reduction Monitor (SRM). Sec. \ref{sec:bmm:hw} details the hardware architecture of the proposed word-wise BMM. Sec. \ref{sec:kyber:ckks} and Sec. \ref{sec:injection} discuss the secret-vector generation processes of Kyber and CKKS, and the corresponding fault-injection methodologies, respectively. Sec. \ref{sec:result} presents the error-detection efficiency and hardware overhead of the proposed fault-detection scheme. Finally, Sec. \ref{sec:con} concludes the paper.
 \begin{algorithm}[!htb]
    \caption{Wordwise BMM}
    \label{algo:bmm}
   \textbf{Input} $a =(a_{l-1},...a_{1}, a_{0})$,
      $b =(b_{l-1},...b{_1}, b_{0})$ \\
    $q =(q_{l-1},...q_{1}, q_{0})$ where $\mu=\lfloor \frac{2^{2\times l}}{q} \rfloor$ \& $k=2\times l$ \\
  \textbf{Output} $c =(c_{l-1},...c_{1}, c_{0}), \rho_1, \rho_2$
    \begin{algorithmic}[1]
    \State $R=0$, $\rho_1=0$, $\rho_2=0$	
      \For{i=0 to $(\frac{l}{w}-1)~~$}
       \For{j=0 to $(\frac{l}{w}-1)~~$}
      \State $a w_i=a_{[iw+w-1...iw]}$
      \State $b w_j=b_{[jw+w-1...iw]}$
      \State $c=a w_i \times b w_j$ \label{line:x1}
      \State $c=c$ || $(i+j) \times w \{0\}$ \label{line:<<1}
      \State $\kappa=(c \times \mu)_{[2k-1... k]}$ \label{line:s1}
      \State $r=c - \kappa\times q $ \label{line:x2}
      \If {($r\geq q$)} \label{line:r1} \Comment{Reduction-1}
        \State $R=R+r-q$,  $\rho_1=1$ \label{line:R}
     \Else
       \State $R=R+r$, $\rho_1=0$
     \EndIf
     \If {$R\geq q$} \label{line:r2} \Comment{Reduction-2}
       \State $R=R-q$, $\rho_2=1$ 
    \Else
       \State $\rho_2=0$
     \EndIf \label{line:e_r} 
      \EndFor    
      \EndFor   
    \State  \textbf{return}   $R$, $\rho_1$, $\rho_2$
    \end{algorithmic}  
    \end{algorithm}

\section{Barrett Modular Multiplication BMM (Wordwise)}
\label{sec:BMM}
The word-wise BMM takes two $l$-bit operands, $a$ and $b$, as inputs, along with an $l$ bit modulus $q$. In each iteration, the algorithm extracts $w$-bit words from the inputs $a$ and $b$. The indices $i$ and $j$ are used to select the corresponding windows from $a$ and $b$, respectively.
As shown in Algorithm~\ref{algo:bmm}, the variables $c$, $\kappa$, and $r$ are used to store intermediate values during each iteration. The variable $c$ stores the shifted partial multiplication result obtained from the product of the selected words from $a$ and $b$. The variable $\kappa$ stores the truncated lower bits of the product between $c$ and the precomputed constant $\mu$, where $\mu = \left\lfloor \frac{2^{2l}}{q} \right\rfloor.$ The variable $r$ stores the intermediate modular reduction result computed as $r = c - \kappa q$. This step provides an approximate reduction of the partial product with respect to the modulus $q$. In line~\ref{line:r1} of Algorithm \ref{algo:bmm}, the first conditional reduction checks whether $r \geq q$. If this condition is satisfied, the modulus $q$ is subtracted from $r$, and the flag $\rho_1$ is set to indicate that a correction was performed. Otherwise, the value of $r$ is directly accumulated into $R$. This first conditional reduction is named as $Reduction-1$ 
Similarly, in line~\ref{line:r2} of Algorithm \ref{algo:bmm}, a second reduction named as $Reduction-2$ is applied to ensure that the accumulated $R$ remains within the valid range $[0, q-1]$. If $R \geq q$, the modulus $q$ is subtracted once more and the flag $\rho_2$ is set to $1$. These $\rho_1$ and $\rho_2$ are the inputs of our statistical reduction monitoring.

\begin{table}[h]
\centering
\caption{Parameters and Description}
\begin{tabular}{|l|p{6.5cm}|}
\hline
\textbf{Parameters} & \textbf{Description} \\
\hline
$c$, $\kappa$ and $r$ & Three internal register of wordwise BMM \\
\hline
$n-1$ & Degree of polynomial \\
\hline
$\lambda$ & Number of faulty loops in the $NTT$, out of $\frac{n}{2} \times \log_{2}(n)$ total loops \\ \hline
$\chi_0$ & Number of faulty loops in the $NTT$ (same total: $\frac{n}{2} \times \log_{2}(n)$) \\ \hline
$\phi$ & Number of faulty bits \\ \hline
$n\chi_0$, $\chi_0$ & Number of operations without any reduction; $n\chi_0$ expressed as a percentage (\%) \\ \hline
$n\chi_1$, $\chi_1$ & Number of operations with Reduction-1 (Line \ref{line:r1} of Algorithm \ref{algo:bmm}); $n\chi_1$ expressed as a percentage (\%) \\ \hline
$n\chi_2$, $\chi_2$ & Number of operations with Reduction-2 (Line \ref{line:r2} of Algorithm \ref{algo:bmm}); $n\chi_2$ expressed as a percentage (\%) \\ \hline
$n\chi_3$, $\chi_3$ & Number of operations with Reduction-1 and Reduction-2 (Line \ref{line:r1} \& Line \ref{line:r2} of Algorithm \ref{algo:bmm}); $n\chi_3$ expressed as a percentage (\%) \\ \hline
\end{tabular}
\label{tab:parameters}
\end{table}

\begin{table}[!t]
\caption{Distribution of Exact and Approximate Quotients for different Left-Shift where l=12, w=4}
\label{tab:quotient_distribution}
\centering
\begin{tabular}{|p{2cm}|p{2cm}|c|c|}
\hline
\textbf{Left Shift (line \ref{line:<<1} of Algo. \ref{algo:bmm})} & \textbf{Possible Values of $c$ at line \ref{line:x1} of Algo. \ref{algo:bmm} } & \textbf{$\boldsymbol{\kappa=\left\lfloor\frac{c}{q}\right\rfloor}$} & \textbf{$\boldsymbol{\kappa=\left\lfloor\frac{c}{q}\right\rfloor-1}$} \\
\hline
$\ll 0$  & 226 & 226 & 0  \\
\hline
$\ll 4$  & 226 & 226 & 0  \\
\hline
$\ll 8$  & 226 & 226 & 0  \\
\hline
$\ll 12$ & 226 & 221 & 5  \\
\hline
$\ll 16$ & 226 & 154 & 72 \\
\hline
\end{tabular}
\end{table}

\section{Fault Detection with Statistical Reduction Monitoring (SRM)}
\label{sec:srm}
Our Statistical Reduction Monitoring (SRM) tracks the frequency of reduction operations during the execution of the algorithm. For the two reduction steps: $Reduction-1$ and $Reduction-2$ described in Line~\ref{line:r1} and Line~\ref{line:r2} of Algorithm~\ref{algo:bmm}, respectively, four possible cases may occur: (i) no reduction, (ii) only the reduction-1, (iii) only the reduction-2, and (iv) both reductions. 
Sec~\ref{sec:expl} explains that when the input operand $a$ is a uniformly distributed random number and $b$ is an $NTT$ twiddle factor (which is generally the case during key generation in lattice-based PQC and FHE schemes) as shown in Algorithm \ref{algo:bmm}, the proposed word-wise $BMM$ most frequently bypasses both reduction stages, referred to as the 'no reduction' case (i). In contrast, the simultaneous execution of both reduction stages, referred to as the 'both reduction' case (iv), occurs with the lowest probability. The 'only reduction-2' case (iii) is the second most frequent execution path, whereas the 'only reduction-1' case (ii) is the third most frequent.
The same statistical execution pattern can be explained when the input operand $a$ in Algorithm~\ref{algo:bmm} corresponds to the intermediate $NTT$ coefficients derived from the secret key during the key generation process of both Kyber and CKKS. The details observation is discussed in Sec. \ref{sec:keygen:kyber} and Sec. \ref{sec:keygen:ckks}.
As shown in Algorithm~\ref{algo:fault:bmm}, the $SRM$ takes $\rho_1$ and $\rho_2$ as inputs form Algorithm \ref{algo:bmm} and computes the percentages $\chi_0$, $\chi_1$, $\chi_2$, and $\chi_3$, corresponding to the four cases mentioned above.
\subsection{Explanation of Probabilities}
\label{sec:expl}
If total $n$ events are there, probability for 'No Reduction' is $P_{\chi0}$=$\frac{n\chi_{0}}{n}$, Probability for ' only the reduction-1' is $P_{\chi1}$=$\frac{n\chi_{1}}{n}$, Probability for ' only the reduction-2' is $P_{\chi2}$=$\frac{n\chi_{2}}{n}$ and finally Probability for 'both reduction' is $P_{\chi3}$=$\frac{n\chi_{3}}{n}$. Here $n$=$n\chi_{0}$+$n\chi_{1}$+$n\chi_{2}$+$n\chi_{3}$. From our experimental observation, the probability relationship without any fault occurrence in $c$, $\kappa$ and $r$ can be expressed as:
\begin{equation}
    P_{\chi0} > P_{\chi2} > P_{\chi1} > P_{\chi3}
\end{equation}

\subsubsection{Probability of Reduction-1 Only}
\label{sec:red1}
The reduction-1 stage calculated by line \ref{line:x2} of Algorithm \ref{algo:bmm} is an intermediate remainder $r$.
\begin{align}
r = c-\kappa q,
\label{eq:r}
\end{align}
where \(c\) is the intermediate shifted product and \(\kappa\) is the approximate quotient derived from the Barrett approximation that satisfies Theorem~\ref{thm:exact_quotient} which shows that only two cases are possible $\kappa=Q$ or $\kappa=Q-1$. If $\kappa=Q$ event occurs then the resulting remainder satisfies $0 \le r <q$ [by Eq.~\ref{case1_1} of Theorem \ref{lem:nonnegative_remainder}]. Then, the reduction-1 stage is never executed. When $\kappa=Q-1$ occurs, then  $q \le r <2q$ [by Eq.~\ref{case2_1} of Theorem \ref{lem:nonnegative_remainder}]. which activates reduction-1 only operation. 
\par Table~\ref{tab:quotient_distribution} presents the distribution of the exact and approximate Barrett quotients for all possible intermediate partial products generated by the proposed word-wise multiplication under the Kyber $NTT$ parameter set. As shown in Algorithm~\ref{algo:bmm}, Line~\ref{line:x1} computes the product of two w-bit words : $aw_i$ and $bw_i$. For w=4, the maximum partial product is $15\times15=225$, therefore, 226 distinct possible values in the range [0,225]. Theorem~\ref{thm:exact_quotient} guarantees that the computed quotient $\kappa$ can only be equal to the exact quotient Q ($\left\lfloor\frac{c}{q}\right\rfloor$) or Q-1($\left\lfloor\frac{c}{q}\right\rfloor-1$). Table~\ref{tab:quotient_distribution} shows that the exact quotient is obtained for all 226 possible partial products when the left-shift amount is 0, 4, or 8 bits in line~\ref{line:<<1} of the Algorithm~\ref{algo:bmm}. For the 12-bit left shift, the exact quotient is produced for 221 out of 226 cases, whereas for the 16-bit left shift, the exact quotient is produced 154 out of 226 cases, with the remaining cases corresponding to the $\left\lfloor\frac{c}{q}\right\rfloor-1$ approximation predicted by Theorem~\ref{thm:exact_quotient}.
Therefore, the Table \ref{tab:quotient_distribution} demonstrates that the exact quotient $\kappa=Q$ produces for most of the intermediate values at every shift level and the event $\kappa=Q-1$ occurs infrequently. Thus, the quotient $\kappa = Q - 1$, which is the necessary condition for activating reduction-1 only, occurs for only a relatively small number of intermediate values. Hence, the probability of triggering reduction-1 is very low.
\begin{align}
P_{\chi1}=P(r\ge q),
\end{align}
remains significantly smaller than the probabilities of the remaining execution paths.

\subsubsection{Probability of Reduction-2 Only} 
The probability of reduction-2 only happened when reduction-1 is skipped. If $\kappa$ is $\left\lfloor\frac{c}{q}\right\rfloor$, then the intermediate remainder satisfies $0 \leq r<q$. Therefore, the Reduction-1 stage is not activated. Since the Barrett approximation yields the exact quotient in the overwhelming majority of executions, this Reduction-1 occurs with a very less probability. Thereafter, both branches of the Reduction~1 stage accumulate the intermediate partial remainder $r$, into the intermediate variable $R$. Therefore, Reduction-2 only depends on the accumulated value $R$ from previous iterations. $R=\sum r_i$. Even of each individual partial remainder satisfy $0 \leq r_i \leq q$. Thus Reduction-1 depends on Barrett approximation error which requires a rare event $\kappa=\left\lfloor\frac{c}{q}\right\rfloor-1$, however, Reduction-2 only depends on normal accumulation of remainders which requires $R+r\geq q$.
This accumulation occurs relatively much more frequently as compared to Barrett quotient underestimation, hence, $P_{\chi2} > P_{\chi1}$.
\subsubsection{Probability of Both Reduction} 
Reduction-1 in Algorithm~\ref{algo:bmm} is triggered when the accumulated intermediate value $R$, generated from the previous iterations, exceeds the modulus $q$. However, the occurrence of Reduction-1 does not necessarily imply the occurrence of Reduction-2. Although Reduction-1 reduces the accumulated value, the resulting intermediate value may still remain below the $q$ required to activate Reduction-2. Consequently, not all executions of Reduction-1 lead to the execution of Reduction-2. As discussed in Section~\ref{sec:red1}, the probability of triggering Reduction-1, denoted by $P_{\chi1}$, is relatively low. Since Reduction-2 can only occur for a subset of the cases in which Reduction-1 is activated, it follows that
$P_{\chi3} \leq P_{\chi1}$.
 \subsubsection{Probability of No Reduction} 
The no-reduction case occurs when none of the reduction events, namely the Reduction-1 only event, the Reduction-2 only event, and the both reductions event, is activated. As discussed in Section~\ref{sec:red1}, \(P_{\chi1}\) is already relatively small, which consequently makes $P_{\chi3}$ even less probable. Furthermore, the intermediate value $R$ is accumulated from $r$. Since, the accumulated value $R$ is, in most cases, smaller than modulus $q$. Therefore, the probability of triggering Reduction-2 only, $P_{\chi2}$, is also relatively low. As a result, a large fraction of executions bypass both reduction stages simultaneously. Since the no-reduction event represents the complement of all reduction events,
\begin{align}
P_{\chi0}
=
1-\left(P_{\chi1}+P_{\chi2}+P_{\chi3}\right),
\end{align}
it exhibits the highest occurrence probability among all execution paths.
\begin{algorithm}[!htb]
\caption{Statistical Reduction Monitoring (SRM)}
\label{algo:fault:bmm}
\textbf{Input} $\rho_1$, $\rho_2$;~~\textbf{Output} $f$
\begin{algorithmic}[1]

\State $n\chi_0=0$, $n\chi_1=0$, $n\chi_2=0$, $n\chi_3=0$
\State $\sigma = NR \times \frac{l}{w} \times \frac{l}{w}$  

\If {($\rho_1=0$) \textbf{and} ($\rho_2=0$)}
    \State $n\chi_0 \gets n\chi_0 + 1$
\EndIf

\If {($\rho_1=1$) \textbf{and} ($\rho_2=0$)}
    \State $n\chi_1 \gets n\chi_1 + 1$
\EndIf

\If {($\rho_1=0$) \textbf{and} ($\rho_2=1$)}
    \State $n\chi_2 \gets n\chi_2 + 1$
\EndIf

\If {($\rho_1=1$) \textbf{and} ($\rho_2=1$)}
    \State $n\chi_3 \gets n\chi_3 + 1$
\EndIf

\If {$check = 1$}

    \For{$i = 0$ to $3$}
        \State $\chi_i \gets \dfrac{n\chi_i \times 100}{\sigma}$
    \EndFor

    \For{$i = 0$ to $3$}
        \If{$\chi_i < \chi_{i,\min}$ \textbf{or} $\chi_i > c_{i,\max}$}
            \State $\chi_{i,\text{fault}} \gets 1$
        \Else
            \State $\chi_{i,\text{fault}} \gets 0$
        \EndIf
    \EndFor

\State $f= \chi_{0,\text{fault}}$ || $\chi_{1,\text{fault}}$ || $\chi_{2,\text{fault}}$ || $\chi_{3,\text{fault}}$

\EndIf
\State \textbf{return} $f$
\end{algorithmic}
\end{algorithm}

\begin{theorem}[Exact Quotient Property of Barrett Reduction]
\label{thm:exact_quotient}
Let $c = m2^{s}$ [in line 7 of Algorithm \ref{algo:bmm}], where $m$ denotes the partial product of two $w$-bit words and $s$ is the corresponding left-shift. Let $Q=\left\lfloor\frac{c}{q}\right\rfloor$ be the exact quotient, and let $\mu=\left\lfloor\frac{2^{k}}{q}\right\rfloor$ be the Barrett precomputed constant, where $k=2\times l$ is the Barrett scaling factor. The approximate quotient computed by Barrett reduction is $\kappa=\left\lfloor\frac{c\mu}{2^{k}}\right\rfloor$ [in line 8 of Algorithm \ref{algo:bmm}]. Then, $\kappa\in\{Q,Q-1\}.$
\end{theorem}

\begin{proof}
Since $\mu=\left\lfloor\frac{2^{k}}{q}\right\rfloor$, there exists an $\delta$, $0 < \delta < 1$  such that $\mu+\delta = \frac{2^{k}}{q}$ .
Therefore,
\begin{equation}
\frac{\mu}{2^{k}}=\frac{1}{q}-\frac{\delta}{2^{k}}.
\end{equation}
Now, substituting this expression into the Barrett approximate quotient  $\kappa=\left\lfloor\frac{c\mu}{2^{k}}\right\rfloor$ gives

\begin{align} 
\kappa &=\left\lfloor{c\left(\frac{1}{q}-\frac{\delta}{2^{k}}\right)}\right\rfloor \\
&= \left\lfloor{\frac{c}{q}-\frac{\delta c}{2^{k}}}\right\rfloor
\end{align}
Let us assume the exact division as
\begin{equation}
\frac{c}{q}=Q+f,
where ~Q=\left\lfloor\frac{c}{q}\right\rfloor and ~f=frac(\frac{c}{q}),~ 0\le f<1.
\end{equation}
Hence,
$\kappa=\left\lfloor{ Q+f-\frac{\delta c}{2^{k}}}\right\rfloor$. Since $0\le\frac{\delta c}{2^{k}}<1$, the subtraction term can reduce the integer part by at most one. Therefore, $\kappa\in\{Q,Q-1\}$.

If $f>\frac{\delta c}{2^{k}}$, then
\begin{equation}
Q<Q+f-\frac{\delta c}{2^{k}}<Q+1,
\end{equation}
which yields $\kappa=Q$.

Otherwise, $f\le\frac{\delta c}{2^{k}}$, which implies $Q-1<Q+f-\frac{\delta c}{2^{k}}<Q$, and therefore $\kappa=Q-1.$

Thus,
\begin{equation}
\kappa=
\begin{cases}
Q,&
\displaystyle
f=\operatorname{frac}\!\left(\frac{c}{q}\right)>
\frac{\delta c}{2^{k}},\\[2ex]
Q-1,
&
\displaystyle
f=\operatorname{frac}\!\left(\frac{c}{q}\right)\le
\frac{\delta c}{2^{k}}.
\end{cases}
\end{equation}

Hence, completes the proof.
\end{proof}

\begin{theorem}
\label{lem:nonnegative_remainder}
Let
$\kappa=\left\lfloor\frac{c\mu}{2^{k}}\right\rfloor$, exact quotient $Q=\left\lfloor\frac{c}{q}\right\rfloor$ 
where $\mu=\left\lfloor\frac{2^{k}}{q}\right\rfloor$ and let the Barrett remainder be computed as $r=c-\kappa q.$
Then the remainder produced by the proposed Barrett modular multiplication algorithm always satisfies  $0 \le r < 2q.$
\end{theorem}

\begin{proof}
Since
$\mu=\left\lfloor\frac{2^{k}}{q}\right\rfloor$
therefore, the floor function gives
\begin{equation}
\mu\le\frac{2^{k}}{q}.
\end{equation}

Now, multiply both sides by the non-negative integer $c$. Then we have, 
\begin{equation}
c\mu\le\frac{c\,2^{k}}{q}.
\end{equation}

Dividing both sides by $2^{k}$ gives
\[
\frac{c\mu}{2^{k}}\le\frac{c}{q}.
\]

Since the floor function preserves inequality,
\[
\kappa=\left\lfloor\frac{c\mu}{2^{k}}\right\rfloor\le\left\lfloor\frac{c}{q}\right\rfloor.
\]

Hence,
\begin{equation}
\left\lfloor\frac{c}{q}\right\rfloor-\kappa\ge0.
\end{equation}

Multiplying both sides by the positive modulus $q$ gives
\begin{equation}
\left(\left\lfloor\frac{c}{q}\right\rfloor-\kappa\right)q\ge0.
\end{equation}

Now, we have
\[
\begin{aligned}
r&=c-\kappa q\\
&=\left(\frac{c}{q}-\kappa\right)q\ge \left(\left\lfloor\frac{c}{q}\right\rfloor-\kappa\right)q \ge 0.
\end{aligned}
\]
Therefore, $r\ge0.$
Thus, the remainder computed by the proposed Barrett modular multiplication algorithm is always non-negative.
\par
Now, from theorem \ref{thm:exact_quotient}, $\kappa \in \{Q, Q-1\}$. Therefore, only two cases need to be considered.\\
\textbf{Case 1:} Exact Quotient ($\kappa=Q$) \\ 
then the remainder becomes $r=c-\kappa q=c-Qq$.
By the Euclidean division theorem, clearly we can write 
\begin{equation}
\label{case1}
    0 \le c-Qq < q, 
\end{equation}
which implies
\begin{equation}
\label{case1_1}
  0 \le r < q   
\end{equation}
\textbf{Case 2:} Approximate Quotient ($\kappa=Q-1$) \\ 
Substituting $\kappa=Q-1$ into the remainder equation gives
$
\begin{aligned}
 r&=c- \kappa q\\
 &=c-(Q-1)q\\ 
 &=(c-Qq)+q  
\end{aligned}
$
\\
Since, $0 \le c-Qq < q$, by the equation \ref{case1}\\
Adding q throughout the inequality gives, $q \le c-Qq +q< 2q$.
Therefore,
\begin{equation}
\label{case2_1}
  q \le r < 2q.  
\end{equation}
\\
Hence, by combining above two cases, the intermediate remainder computed by the proposed Barrett modular multiplication is always satisfies $0 \le r < 2q$.

\end{proof}

\section{Hardware Architecture of SRM Protected $BMM$}
\label{sec:bmm:hw}
As shown in Fig. \ref{fig:arch_bmm}, the proposed $BMM$ uses 2 multipliers $X_1$ and $X_2$ where $X_1$ is used to compute $c$ by multiplying $aw_i$ and $bw_j$ (line \ref{line:x1} of Algorithm \ref{algo:bmm}). The $X_2$ multiplier is used to compute $\kappa$ by multiplying $c$ and $\mu$ (line \ref{line:x2} of Algorithm \ref{algo:bmm}).Among the two subtraction blocks, $-_1$ and $-_2$, $-_1$ is used to compute $r$ at line \ref{line:s1} of Algorithm \ref{algo:bmm}, whereas $-_2$ is used to compute Reduction-1 and Reduction-2. The subtractor $-_2$ also computes $\rho_1$ and $\rho_2$, where $\rho_1$ is set to 1 if the Reduction-1 condition is satisfied; otherwise, it is set to 0. Similarly, $\rho_2$ is asserted if the Reduction-2 condition is satisfied; otherwise, it remains deasserted. 
The $\chi$~$gen$ block takes $\rho_1$ and $\rho_2$ as input and run Algorithm \ref{algo:fault:bmm} and calculate $\chi_0$, $\chi_1$ and $\chi_2$. Although the calculations of $\chi_0$, $\chi_1$, and $\chi_2$ in Algorithm \ref{algo:fault:bmm} are expressed as percentages (base 100), the actual $SRM$ implementation uses 128 instead of 100 to avoid division operations. This is because division by 100 is computationally expensive in hardware, whereas division by 128 can be efficiently implemented using a 7-bit shift operation. From the experimental results reported in Table \ref{tab:kyber:c}, Table \ref{tab:kyber:q}, Table \ref{tab:kyber:r}, Table \ref{tab:ckks:c}, Table \ref{tab:ckks:q}, and Table \ref{tab:ckks:r}, it is observed that $\chi_3$ cannot exclusively detect any fault occurrence that is not already detected by $\chi_0$, $\chi_1$, or $\chi_2$. Therefore, to reduce hardware overhead, the $SRM$ omits the computation of $\chi_3$.Finally, the $matcher$ block takes $\chi_0$, $\chi_1$, and $\chi_2$ from the $\chi$~$gen$ block and flags the final fault occurrence at $f$.
This proposed $BMM$ has two pipeline stages. In the first pipeline stage, the computation is performed up to the shifted $c$, as stated in line \ref{line:<<1} of Algorithm \ref{algo:bmm}. In the second pipeline stage, the final computation of $R$ is performed through two reduction conditions.
\begin{figure}[!htb]
\centering
\includegraphics[width=0.5\textwidth]{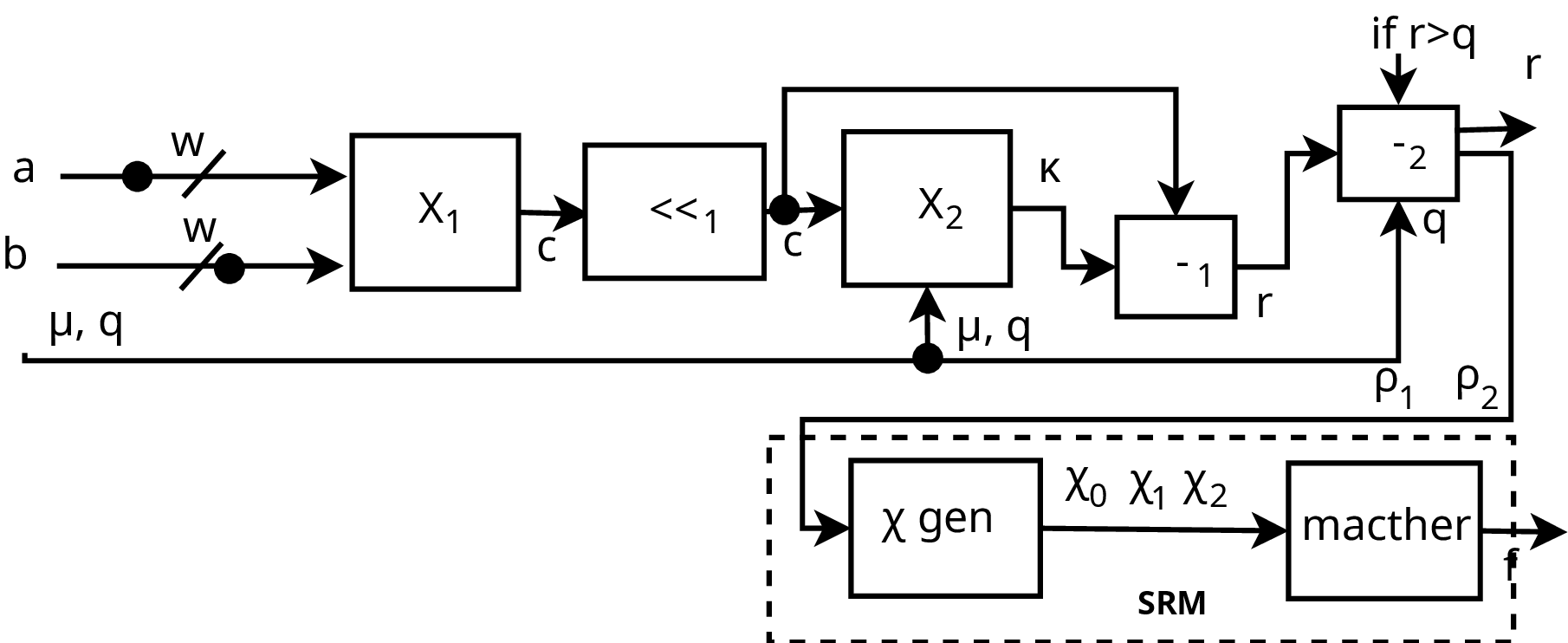}
\vspace{-5pt}
\caption{Hardware Architecture of Wordwise BMM}
\vspace{-5pt}
\label{fig:arch_bmm}
\end{figure}    
\section{Secret Vectors in $NTT$ for Kyber \& CKKS}
\label{sec:kyber:ckks}
In this paper, we have studied two cases to validate our proposed fault detection approach using a multi-conditional reduction path: secret vector generation in the CRYSTALS-Kyber PQC algorithm and the CKKS FHE algorithm. Both cases validate the probability distribution of the four reduction situations explained in Sec.~\ref{sec:expl}.
\subsection{Secret Vector in Kyber}
\label{sec:keygen:kyber}
The three Kyber variants, namely Kyber-512, Kyber-768, and Kyber-1024, consist of three main processes: Key Generation, Encryption, and Decryption. Let us consider the example of Kyber-768, which requires 3 $NTT$ operations on the secret vector $s$ and 3 $NTT$ operations on the error polynomial $e$. Similarly, the encryption and decryption processes also require multiple $NTT$ and inverse $NTT$ operations. Our word-wise $BMM$ is the most resource, time, and power critical operation among all these $NTT$ computations. Although $BMM$ is used in different stages of Kyber key generation, encryption, and decryption, this paper discusses the secret vector generation process in greater detail.
\par As per FIPS203 \cite{FIPS203}, these secret polynomials are generated using specific PRF and CBD processes, where the secret vector $s$ is sampled in a uniform and random manner. As shown in Algorithm~\ref{algo:bmm}, the word-wise $BMM$ used inside the $NTT$ takes $s$ as input $a$ and the twiddle factor $\omega$ as input $b$. This forward $NTT$ transformation is shown in line \ref{line:s} of Kyber Key Generation, shown in  Algorithm~\ref{algo:kyber:keygen}. Our SR method monitors the behavior of the $NTT$ operation running with the word-wise $BMM$ by tracking how many times it executes without Reduction~1 and Reduction~2 ($n\chi_0$), with only Reduction~1 ($n\chi_1$), with only Reduction~2 ($n\chi_2$), and with both Reduction~1 and Reduction~2 ($n\chi_3$). The percentage values of $\chi_0$, $\chi_1$, $\chi_2$, and $\chi_3$ remain within the ranges of $77.66$--$80.41$, $0.14$--$0.61$, $19.18$--$21.91$, and $0.00$--$0.24$, respectively, when no fault occurs throughout the 1024 $NTT$ iterations. 
\begin{algorithm}[!htbp]
\caption{Kyber.CPA.PKE.KeyGen()}
\label{algo:kyber:keygen}
\begin{algorithmic}[1]
\State \textbf{Output:} Secret key $sk \in \mathbb{B}^{12 \cdot k \cdot n / 8}$
\State \textbf{Output:} Public key $pk \in \mathbb{B}^{12 \cdot k \cdot n / 8 + 32}$
\State $d \leftarrow \mathbb{B}^{32}$
\State $(\rho, \sigma) := G(d)$
\State $N := 0$
\For{$i = 0$ \textbf{to} $k-1$}
  \For{$j = 0$ \textbf{to} $k-1$}
    \State $\hat{A}[i][j] := \text{Parse}(XOF(\rho, j, i))$
  \EndFor
\EndFor
\For{$i = 0$ \textbf{to} $k-1$}
  \State $s[i] := \text{CBD}_{\eta_1}(\text{PRF}(\sigma, N))$ \label{line:s}
  \State $N := N + 1$
\EndFor
\For{$i = 0$ \textbf{to} $k-1$}
  \State  $e[i] := \text{CBD}_{\eta_1}(\text{PRF}(\sigma, N))$ \label{line:e}
  \State $N := N + 1$
\EndFor
\State $\hat{s} := \text{NTT}(s)$
\State $\hat{e} := \text{NTT}(e)$
\State $\hat{t} := \hat{A} \circ \hat{s} + \hat{e}$
\State $pk := (E_{12}(\hat{t} \bmod q) \| \rho)$
\State $sk := E_{12}(\hat{s} \bmod q)$
\Return $(pk, sk)$
\end{algorithmic}
\end{algorithm}

\subsection{Secret Vector in CKKS}
\label{sec:keygen:ckks}
As mentioned in Section~\ref{sec:keygen:kyber}, the key generation process of the \textit{CKKS} algorithm has similar requirements for the word-wise $BMM$ operation inside its $NTT$ computation. The primary difference is that the word-wise $BMM$ used in the CKKS $NTT$ operates with different polynomial degrees $n-1$ and moduli $q$. The secret vector generation $s$ of CKKS is generated from a ternary or Gaussian distribution. As shown in line \ref{line:ckks:s} of Algorithm~\ref{algo:ckks:keygen}, $s$ is passed through an $NTT$ operation. Similar to the Kyber variants, our wordwise $BMM$ serves as the fundamental building block of this $NTT$ operation, which is also the most power- and resource-critical component. Our SR method monitors the behavior of the $NTT$ operation running with the worwise $BMM$ by tracking Reduction-1 and Reduction-2. The percentages corresponding to the four possible conditions, namely no reduction $\chi_0$, only Reduction-1 $\chi_1$, only Reduction-2 $\chi_2$, and both reductions 1 and 2 $\chi_3$, range from 80.06\%--83.81\%, 0.07\%--0.11\%, 16.11\%--19.85\%, and 0.00\%--0.01\%, respectively.
\begin{algorithm}[H]
\caption{CKKS.KeyGen()}
\label{algo:ckks:keygen}
\begin{algorithmic}[1]

\State \textbf{Input:} Ring dimension $N$, modulus $q$
\State \textbf{Output:} Secret key $sk$
\State \textbf{Output:} Public key $pk$

\State $a \leftarrow R_q$
\State $s \leftarrow \chi^{N}$ \label{line:ckks:s}
\State $e \leftarrow \chi^{N}$ \label{line:ckks:e}

\State $\hat{s} := \text{NTT}(s)$ \label{line:ckks:s}
\State $\hat{e} := \text{NTT}(e)$
\State $\hat{a} := \text{NTT}(a)$

\State $\hat{b} := -(\hat{a} \circ \hat{s} + \hat{e}) \bmod q$

\State $pk := (\hat{b}, \hat{a})$
\State $sk := \hat{s}$

\Return $(pk, sk)$

\end{algorithmic}
\end{algorithm}

\section{Fault Injection in $NTT$ for Kyber \& CKKS}
\label{sec:injection}
In this paper, we study two types of fault injection on the $NTT$ intermediate variables $c$, $\kappa$, and $r$: random bit-flip faults and burst bit-flip faults. In random fault injection, bits are flipped from $0$ to $1$ at random positions in $c$, $\kappa$, and $r$, whereas in burst fault injection, consecutive bit positions are flipped in these variables. For an $(n-1)$-degree polynomial, the $NTT$ requires $\frac{n}{2}\log n$ butterfly iterations, and in each iteration, the $BMM$ is invoked. For our test cases, Kyber uses $n=256$ and CKKS uses $n=4096$, resulting in $1024$ and $24576$ $NTT$ iterations, respectively. Random and burst faults are injected during each $NTT$ iteration for both Kyber and CKKS. Table \ref{tab:parameters} presents the description of all parameters used in this study. Here, $\lambda$ represents the number of faulty $NTT$ iterations, and $\phi$ denotes the number of faulty bits in the faulty $NTT$ loop. 
As shown in Table \ref{tab:kyber:c}, Table \ref{tab:kyber:q} and Table \ref{tab:kyber:r}, faults are injected into the three intermediate registers $c$, $\kappa$, and $r$, respectively, used in the $BMM$ of Kyber’s $NTT$ by varying $\lambda$ and $\phi$. It is observed that up to 128 faulty $NTT$ iterations, with any number of injected faults, cause the values of $\chi_0$, $\chi_1$, and $\chi_2$ to deviate from their corresponding values obtained when no faults occur in any $NTT$ loop. Therefore, if the fault persists for at least 128 $NTT$ iterations in Kyber’s $NTT$, our model can detect any number of faulty bits with 100\% efficiency. When the number of faulty $NTT$ iterations is $64$ or fewer, some values of $\chi_0$, $\chi_1$, and $\chi_2$ may coincide with their corresponding values obtained when no faults occur in any $NTT$ loop. Therefore, beyond $64$ faulty iterations during the secret vector computation process in Kyber’s $NTT$, the fault detection efficiency decreases from 100\%. The detailed values of $\chi_0$, $\chi_1$, $\chi_2$ and $\chi_3$ for the secret vector computation process in Kyber $NTT$ are shown in Table \ref{tab:kyber:c}, Table \ref{tab:kyber:q} and Table \ref{tab:kyber:r}.

\begin{table}[!htb]
\centering
\begin{tabular}{|c|c|c|c|c|c|}
\hline
$\lambda$ & $\phi$ & $\chi_0$ (\%) & $\chi_1$ (\%) & $\chi_2$ (\%) & $\chi_3$ (\%) \\ \hline
$\times$&0&\textbf{77.66--80.41} &\textbf{0.14--0.61}&\textbf{19.18--21.91}&\textbf{0.00, 0.24}\\\hline

\multicolumn{6}{|c|}{\textbf{Random Fault in c}} \\ \hline

1024&1&\cellcolor{gray}{56.11--58.72} &\cellcolor{gray}{2.31--3.57}&\cellcolor{gray}{38.24--40.87}&0.14, 0.60\\
512&1&\cellcolor{gray}{65.49--68.58} &\cellcolor{gray}{1.16--2.22}&\cellcolor{gray}{29.61--32.76}&0.05, 0.37\\
128&1&\cellcolor{gray}{73.81--76.81} &\cellcolor[gray]{0.8}{0.41--1.09}&\cellcolor{gray}{22.41--25.55}&0.00, 0.26\\
64&1&75.42--78.50 &0.29--0.82&20.93--23.87&0.00, 0.25\\\hline
1024&2&\cellcolor{gray}{48.81--51.43}&\cellcolor{gray}{5.10--6.88}&\cellcolor{gray}{41.80--44.39}&\cellcolor{gray}{0.47, 1.17}\\
512&2&\cellcolor{gray}{61.39--64.89} &\cellcolor{gray}{2.47--3.99}&\cellcolor{gray}{31.62--34.67}&0.17, 0.68\\
128&2&\cellcolor{gray}{72.79--75.97} &\cellcolor{gray}{0.73--1.55}&\cellcolor{gray}{22.91--26.12}&0.02, 0.35\\
64&2&74.87--78.09 &0.47--1.04&\cellcolor[gray]{0.8}{21.14--24.28}&0.00, 0.27\\\hline
1024&3&\cellcolor{gray}{45.66--48.32} &\cellcolor{gray}{7.44--9.60}&\cellcolor{gray}{41.95--44.43}&\cellcolor{gray}{0.81, 1.63}\\
512&3&\cellcolor{gray}{60.12--63.30} &\cellcolor{gray}{3.73--5.28}&\cellcolor{gray}{31.52--34.71}&\cellcolor{gray}{0.36, 0.97}\\
128&3&\cellcolor{gray}{72.41--75.42} &\cellcolor{gray}{1.01--1.89}&\cellcolor{gray}{22.91--26.00}&0.07, 0.42\\
64&3&\cellcolor[gray]{0.8}{74.91--77.97} &\cellcolor[gray]{0.8}{0.51--1.28}&\cellcolor[gray]{0.8}{21.08--24.16}&0.02, 0.33\\\hline

\multicolumn{6}{|c|}{\textbf{Burst Fault in c}} \\ \hline

1024&1&\cellcolor{gray}{55.97--58.66} &\cellcolor{gray}{2.30--3.59}&\cellcolor{gray}{38.19--40.74}&0.12, 0.56\\
512&1&\cellcolor{gray}{65.56--68.57}&\cellcolor{gray}{0.98--2.18}&\cellcolor{gray}{29.54--32.65}&0.04, 0.37\\
128&1&\cellcolor{gray}{73.63--76.89}&0.42--1.06&\cellcolor{gray}{22.34--25.69}&0.00, 0.27\\
64&1&75.41--78.49 &0.28--0.88&20.91--23.99&0.00, 0.24\\\hline
1024&2&\cellcolor{gray}{57.16--59.95} &\cellcolor{gray}{2.55--3.98}&\cellcolor{gray}{36.30--39.13}&\cellcolor{gray}{0.28, 0.85}\\
512&2&\cellcolor{gray}{66.02--69.05} &\cellcolor{gray}{1.25--2.31}&\cellcolor{gray}{28.80--31.92}&0.10, 0.51\\
128&2&\cellcolor{gray}{73.85--77.03} &0.42--1.12&\cellcolor{gray}{21.98--25.25}&0.01, 0.27\\
64&2&75.53--78.83 &0.29--0.86&20.54--23.85&0.00, 0.26\\\hline
1024&3&\cellcolor{gray}{57.86--60.66} &\cellcolor{gray}{2.51--4.04}&\cellcolor{gray}{35.53--38.17}&\cellcolor{gray}{0.36, 1.03}\\
512&3&\cellcolor{gray}{66.60--69.33} &\cellcolor{gray}{1.30--2.35}&\cellcolor{gray}{28.26--31.22}&0.13, 0.60\\
128&3&\cellcolor{gray}{74.21--77.48} &0.44--1.07&\cellcolor[gray]{0.8}{21.61--24.96}&0.00, 0.31\\
64&3&75.81--78.54 &0.28--0.87&20.79--23.57&0.00, 0.28\\\hline
\multicolumn{6}{|c|}{$\lambda:$ \# Faulty Loops, $\phi:$ \# Faulty bits, l=12, w=4, q=3329, n=256} \\ \hline
\end{tabular}
\caption{Random and Burst Fault in c for Kyber }
\label{tab:kyber:c}
\end{table}

\begin{table}[h]
\centering
\begin{tabular}{|c|c|c|c|c|c|}
\hline
$\lambda$ & $\phi$ & $\chi_0$ (\%) & $\chi_1$ (\%) & $\chi_2$ (\%) & $\chi_3$ (\%) \\ \hline
$\times$&0&\textbf{77.66--80.41} &\textbf{0.14--0.61}&\textbf{19.18--21.91}&\textbf{0.00, 0.24}\\\hline
\multicolumn{6}{|c|}{\textbf{Random Fault in $\kappa$}} \\ \hline
1024&1&\cellcolor{gray}{94.64--96.46}&\cellcolor{gray}{3.54--5.36}&\cellcolor{gray}{0.00--0.02}&0.00, 0.04\\
512&1&\cellcolor{gray}{85.83--88.56} &\cellcolor{gray}{1.81--2.93}&\cellcolor{gray}{9.17--11.59}&0.00, 0.15\\
128&1&\cellcolor[gray]{0.8}{79.73--82.51}&\cellcolor[gray]{0.8}{0.52--1.23}&\cellcolor[gray]{0.8}{16.45--19.35}&0.00, 0.21\\
64&1&78.67--81.60 &0.33--0.91&17.73--20.68&0.00, 0.22\\\hline
1024&2&\cellcolor{gray}{98.43--99.20} &\cellcolor{gray}{0.80--1.57}&\cellcolor{gray}{0.00--0.00}&0.00, 0.01\\
512&2&\cellcolor{gray}{87.48--90.30}&0.49--1.15&\cellcolor{gray}{9.07--11.70}&0.00, 0.17\\
128&2&\cellcolor[gray]{0.8}{9.93--82.91} &0.15--0.75&16.51--19.57&0.00, 0.23\\
64&2&78.85--81.64 &0.15--0.66&17.84--20.67&0.00, 0.26\\\hline
1024&3&\cellcolor{gray}{99.37--99.86} &0.14--0.63&\cellcolor{gray}{0.00--0.00}&0.00, 0.00\\
512&3&\cellcolor{gray}{88.21--90.52} &0.17--0.61&\cellcolor{gray}{9.09--11.44}&0.00, 0.16\\
128&3&\cellcolor[gray]{0.8}{80.11--83.37} &0.15--0.59&16.16--19.37&0.00, 0.24\\
64&3&78.78--81.96 &0.17--0.62&17.63--20.82&0.00, 0.25\\\hline
\multicolumn{6}{|c|}{\textbf{Burst Fault in $\kappa$}} \\ \hline
1024&1&\cellcolor{gray}{94.76--96.48} &\cellcolor{gray}{3.50--5.24}&\cellcolor{gray}{0.00--0.02}&0.00, 0.05\\
512&1&\cellcolor{gray}{86.04--88.74} &\cellcolor{gray}{1.79--2.99}&\cellcolor{gray}{9.03--11.66}&0.00, 0.17\\
128&1&\cellcolor[gray]{0.8}{79.43--82.55}&0.48--1.23&\cellcolor[gray]{0.8}{16.50--19.70}&0.00, 0.23\\
64&1&78.74--81.60 &0.31--0.93&17.85--20.54&0.00, 0.27\\\hline
1024&2&\cellcolor{gray}{95.50--97.09} &\cellcolor{gray}{2.91--4.50}&\cellcolor{gray}{0.00--0.01}&0.00, 0.03\\
512&2&\cellcolor{gray}{86.48--88.89} &\cellcolor{gray}{1.46--2.64}&\cellcolor{gray}{8.97--11.48}&0.00, 0.16\\
128&2&\cellcolor[gray]{0.8}{79.62--82.70} &0.43--1.12&\cellcolor[gray]{0.8}{16.33--19.57}&0.00, 0.22\\
64&2&78.63--81.62 &0.30--0.87&17.84--20.71&0.00, 0.23\\\hline
1024&3&\cellcolor{gray}{96.20--97.50} &\cellcolor{gray}{2.50--3.80}&\cellcolor{gray}{0.00--0.01}&0.00, 0.02\\
512&3&\cellcolor{gray}{86.70--89.28} &\cellcolor{gray}{1.26--2.24}&\cellcolor{gray}{9.00--11.50}&0.00, 0.17\\
128&3&\cellcolor[gray]{0.8}{79.73--82.77} &0.41--1.07&\cellcolor[gray]{0.8}{16.46--19.48}&0.00, 0.23\\
64&3&78.53--81.61 &0.26--0.85&17.78--20.86&0.00, 0.24\\\hline

\multicolumn{6}{|c|}{$\lambda:$ \# Faulty Loops, $\phi:$ \# Faulty bits, l=12, w=4, q=3329, n=256} \\ \hline
\end{tabular}
\caption{Random and Burst Fault in $\kappa$ for Kyber}
\label{tab:kyber:q}
\end{table}

\begin{table}[!htbp]
\centering
\begin{tabular}{|c|c|c|c|c|c|}
\hline
$\lambda$ & $\phi$ & $\chi_0$ (\%) & $\chi_1$ (\%) & $\chi_2$ (\%) & $\chi_3$ (\%) \\ \hline

$\times$&0&\textbf{77.66--80.41} &\textbf{0.14--0.61}&\textbf{19.18--21.91}&\textbf{0.00, 0.24}\\\hline

\multicolumn{6}{|c|}{\textbf{Random Fault in r}} \\ \hline
1024&1&\cellcolor{gray}{69.79--72.03} &\cellcolor{gray}{2.46--3.88}&\cellcolor{gray}{24.59--26.79}&0.13, 0.54\\
512&1&\cellcolor{gray}{72.73--75.20}&\cellcolor{gray}{1.29--2.38}&\cellcolor{gray}{22.75--25.29}&0.03, 0.36\\
128&1&75.63--78.65 &0.46--1.11&20.52--23.56&0.00, 0.25\\
64&1&76.52--79.38 &0.30--0.86&20.07--22.75&0.00, 0.24\\\hline
1024&2&\cellcolor{gray}{64.52--67.05}&\cellcolor{gray}{4.24--6.21}&\cellcolor{gray}{26.88--29.67}&\cellcolor{gray}{0.31, 0.90}\\
512&2&\cellcolor{gray}{70.08--72.76} &\cellcolor{gray}{2.22--3.46}&\cellcolor{gray}{24.14--26.73}&0.11, 0.53\\
128&2&74.80--77.99 &\cellcolor{gray}{0.61--1.39}&20.88--24.03&0.00, 0.29\\
64&2&76.11--78.96 &0.37--0.97&20.24--23.12&0.00, 0.26\\\hline
1024&3&\cellcolor{gray}{59.47--62.23}&\cellcolor{gray}{5.77--8.00}&\cellcolor{gray}{30.09--32.66}&\cellcolor{gray}{0.47, 1.22}\\
512&3&\cellcolor{gray}{67.33--70.05}&\cellcolor{gray}{2.96--4.41}&\cellcolor{gray}{25.80--28.91}&0.20, 0.68\\
128&3&\cellcolor{gray}{74.25--77.50}&\cellcolor{gray}{0.81--1.66}&\cellcolor[gray]{0.8}{21.28--24.41}&0.02, 0.34\\
64&3&75.72--78.95 &0.44--1.16&20.18--23.34&0.01, 0.27\\\hline
\multicolumn{6}{|c|}{\textbf{Burst Fault in r}} \\ \hline

1024&1&\cellcolor{gray}{69.56--72.06}&\cellcolor{gray}{2.53--3.85}&\cellcolor{gray}{24.56--26.95}&0.14, 0.56\\
512&1&\cellcolor{gray}{72.82--75.50}&\cellcolor{gray}{1.36--2.38}&\cellcolor{gray}{22.55--25.22}&0.04, 0.35\\
128&1&75.79--78.58 &0.44--1.12&20.56--23.36&0.00, 0.27\\
64&1&76.56--79.43 &0.28--0.88&19.90--22.80&0.00, 0.26\\\hline
1024&2&\cellcolor{gray}{66.61--69.10}&\cellcolor{gray}{2.93--4.44}&\cellcolor{gray}{26.82--29.22}&\cellcolor[gray]{0.8}{0.20, 0.71}\\
512&2&\cellcolor{gray}{70.97--73.67}&\cellcolor{gray}{1.53--2.57}&\cellcolor{gray}{23.97--26.76}&0.07, 0.44\\
128&2&75.21--78.14 &0.48--1.20&20.76--23.95&0.00, 0.29\\
64&2&76.17--79.21 &0.30--0.94&20.02--23.07&0.00, 0.23\\\hline
1024&3&\cellcolor{gray}{64.61--67.36}&\cellcolor{gray}{6.86--8.95}&\cellcolor{gray}{24.49--26.87}&\cellcolor{gray}{0.26, 0.80}\\
512&3&\cellcolor{gray}{69.98--72.75} &\cellcolor{gray}{3.47--5.15}&\cellcolor{gray}{22.63--25.28}&0.10, 0.49\\
128&3&75.07--77.99 &\cellcolor{gray}{0.98--1.86}&20.49--23.35&0.01, 0.30\\
64&3&75.75--79.34 &\cellcolor{gray}{0.56--1.33}&19.67--23.30&0.00, 0.26\\\hline
\multicolumn{6}{|c|}{$\lambda:$ \# Faulty Loops, $\phi:$ \# Faulty bits, l=12, w=4, q=3329, n=256} \\ \hline
\end{tabular}
\caption{Random and Burst Fault in r for Kyber }

\label{tab:kyber:r}
\end{table}

Similarly, the same study is conducted for the secret vector computation in CKKS’s $NTT$, as shown in Tables \ref{tab:ckks:c}, \ref{tab:ckks:q}, and \ref{tab:ckks:r}, by varying the number of faulty loops $\lambda$, the number of faulty bits $\phi$, and the fault locations at $c$, $\kappa$, and $r$, respectively. Results show that up to 512 faulty $NTT$ iterations in CKKS, with any number of injected faults, cause the values of $\chi_0$, $\chi_1$, and $\chi_2$ to deviate from their corresponding values obtained when no faults occur in any $NTT$ loop. Therefore, if the fault persists for at least 512 $NTT$ iterations in CKKS $NTT$, our model can detect any number of faulty bits with 100\% efficiency. When the number of faulty $NTT$ iterations is $512$ or fewer, some values of $\chi_0$, $\chi_1$, and $\chi_2$ may coincide with their corresponding values obtained when no faults occur in any $NTT$ loop. Therefore, beyond $512$ faulty iterations during the secret vector computation process in CKKS $NTT$, the fault detection efficiency decreases from 100\%. The detailed values of $\chi_0$, $\chi_1$, $\chi_2$ and $\chi_3$ for the secret vector computation process in CKKS $NTT$ are shown in Table \ref{tab:ckks:c}, Table \ref{tab:ckks:q} and Table \ref{tab:ckks:r}.

\begin{table}[!htb]
\centering
\begin{tabular}{|c|c|c|c|c|c|}
\hline
$\lambda$ & $\phi$ & $\chi_0$ (\%) & $\chi_1$ (\%) & $\chi_2$ (\%) & $\chi_3$ (\%) \\ \hline

$\times$&0&\textbf{80.06--83.81} &\textbf{0.07--0.11}&\textbf{16.11--19.85}&\textbf{0.00, 0.01}\\\hline
\multicolumn{6}{|c|}{\textbf{Random Fault in c}} \\ \hline
24576&1&\cellcolor{gray}{59.07--59.45}&\cellcolor{gray}{0.64--0.74}&\cellcolor{gray}{39.79--40.18}&\cellcolor{gray}{0.04, 0.07}\\
1024&1&\cellcolor{gray}{74.61--76.59}&\cellcolor{gray}{0.12--0.16}&\cellcolor{gray}{23.28--25.25}&0.00, 0.01\\
512&1&\cellcolor{gray}{76.42--78.66}&0.09--0.14&\cellcolor{gray}{21.22--23.46}&0.00, 0.01\\
128&1&78.60--81.63 &0.08--0.12&18.29--21.29&0.00, 0.01\\\hline
24576&2&\cellcolor{gray}{54.04--54.45}&\cellcolor{gray}{1.36--1.53}&\cellcolor{gray}{44.00--44.42}&\cellcolor{gray}{0.11, 0.15}\\
1024&2&\cellcolor{gray}{74.39--76.34}&\cellcolor{gray}{0.14--0.19}&\cellcolor{gray}{23.49--25.41}&0.00, 0.01\\
512&2&\cellcolor{gray}{76.20--78.72}&\cellcolor[gray]{0.8}{0.11--0.16}&\cellcolor{gray}{21.14--23.65}&0.00, 0.01\\
128&2&78.56--81.72 &0.09--0.13&18.18--21.32&0.00, 0.01\\\hline
24576&3&\cellcolor{gray}{51.60--51.97}&\cellcolor{gray}{2.20--2.40}&\cellcolor{gray}{45.51--45.89}&\cellcolor{gray}{0.19, 0.24}\\
1024&3&\cellcolor{gray}{74.24--76.11} &\cellcolor{gray}{0.18--0.24}&\cellcolor{gray}{23.67--25.54}&\cellcolor[gray]{0.8}{0.01, 0.02}\\
512&3&\cellcolor{gray}{76.17--78.73}&\cellcolor{gray}{0.13--0.18}&\cellcolor{gray}{21.11--23.66}&0.00, 0.01\\
128&3&78.31--81.60 &0.09--0.13&18.30--21.57&0.00, 0.01\\\hline

\multicolumn{6}{|c|}{\textbf{Burst Fault in c}} \\ \hline
24576&1&\cellcolor{gray}{59.08--59.46}&\cellcolor{gray}{0.64--0.74}&\cellcolor{gray}{39.79--40.18}&\cellcolor{gray}{0.04, 0.07}\\
1024&1&\cellcolor{gray}{74.73--76.50} &\cellcolor{gray}{0.11--0.16}&\cellcolor{gray}{23.38--25.12}&0.00, 0.01\\
512&1&\cellcolor{gray}{76.45--78.75}&\cellcolor[gray]{0.8}{0.10--0.14}&\cellcolor{gray}{21.13--23.41}&0.00, 0.01\\
128&1&78.75--81.80 &0.08--0.12&18.11--21.14&0.00, 0.01\\\hline
24576&2&\cellcolor{gray}{54.05--54.40}&\cellcolor{gray}{1.37--1.51}&\cellcolor{gray}{44.00--44.39}&\cellcolor{gray}{0.11, 0.15}\\
1024&2&\cellcolor{gray}{74.32--76.23}&\cellcolor{gray}{0.14--0.19}&\cellcolor{gray}{23.59--25.49}&0.00, 0.01\\
512&2&\cellcolor{gray}{76.36--78.80}&\cellcolor[gray]{0.8}{0.11--0.16}&\cellcolor{gray}{21.06--23.49}&0.00, 0.01\\
128&2&78.66--81.82 &0.08--0.13&18.08--21.23&0.00, 0.01\\\hline
24576&3&\cellcolor{gray}{51.60--51.97}&\cellcolor{gray}{2.22--2.39}&\cellcolor{gray}{45.52--45.87}&\cellcolor{gray}{0.19, 0.24}\\
1024&3&\cellcolor{gray}{74.20--76.06}&\cellcolor{gray}{0.18--0.24}&\cellcolor{gray}{23.73--25.57}&0.01, 0.02\\
512&3&\cellcolor{gray}{76.28--78.84}&\cellcolor{gray}{0.13--0.18}&\cellcolor{gray}{21.01--23.57}&0.00, 0.01\\
128&3&78.65--81.69 &0.09--0.13&18.20--21.23&0.00, 0.01\\\hline
\multicolumn{6}{|c|}{$\lambda:$ \# Faulty Loops, $\phi:$ \# Faulty bits, l=32, w=8, q=1811939329, n=4096} \\ \hline
\end{tabular}
\caption{Random and Burst Fault in c for CKKS}
\label{tab:ckks:c}
\end{table}

\begin{table}[!htb]
\centering
\begin{tabular}{|c|c|c|c|c|c|}
\hline
$\lambda$ & $\phi$ & $\chi_0$ (\%) & $\chi_1$ (\%) & $\chi_2$ (\%) & $\chi_3$ (\%) \\ \hline
$\times$&0&\textbf{80.06--83.81} &\textbf{0.07--0.11}&\textbf{16.11--19.85}&\textbf{0.00, 0.01}\\\hline
\multicolumn{6}{|c|}{\textbf{Random Fault in $\kappa$}} \\ \hline
24576&1&\cellcolor{gray}{59.08--59.44}&\cellcolor{gray}{0.63--0.75}&\cellcolor{gray}{39.80--40.19}&\cellcolor{gray}{0.04, 0.07}\\
1024&1&\cellcolor{gray}{74.66--76.58} &\cellcolor{gray}{0.12--0.16}&\cellcolor{gray}{23.28--25.20}&0.00, 0.01\\
512&1&\cellcolor{gray}{76.49--78.83} &\cellcolor{gray}{0.10--0.14}&\cellcolor{gray}{21.06--23.39}&0.00, 0.01\\
128&1&78.79--81.80 &0.08--0.12&18.11--21.10&0.00, 0.01\\\hline
24576&2&\cellcolor{gray}{54.02--54.42} &\cellcolor{gray}{1.37--1.53}&\cellcolor{gray}{44.01--44.42}&\cellcolor{gray}{0.11, 0.15}\\
1024&2&\cellcolor{gray}{74.47--76.20} &\cellcolor{gray}{0.14--0.20}&\cellcolor{gray}{23.62--25.36}&0.00, 0.01\\
512&2&\cellcolor{gray}{76.33--78.63} &0.11--0.16&\cellcolor{gray}{21.23--23.52}&0.00, 0.01\\
128&2&78.59--81.77 &0.09--0.13&18.13--21.30&0.00, 0.01\\\hline
24576&3&\cellcolor{gray}{51.59--51.97}&\cellcolor{gray}{2.22--2.41}&\cellcolor{gray}{45.49--45.89}&\cellcolor{gray}{0.19, 0.24}\\
1024&3&\cellcolor{gray}{74.23--76.26} &\cellcolor{gray}{0.18--0.24}&\cellcolor{gray}{23.52--25.55}&0.01, 0.02\\
512&3&\cellcolor{gray}{76.10--78.70}&\cellcolor{gray}{0.13--0.18}&\cellcolor{gray}{21.14--23.73}&0.00, 0.01\\
128&3&78.76--81.73 &0.09--0.13&18.17--21.12&0.00, 0.01\\\hline
\multicolumn{6}{|c|}{\textbf{Burst Fault in $\kappa$}} \\ \hline
24576&1&\cellcolor{gray}{59.09--59.45}&\cellcolor{gray}{0.64--0.74}&\cellcolor{gray}{39.80--40.17}&0.04, 0.08\\
1024&1&\cellcolor{gray}{74.68--76.76}&0.11--0.16&\cellcolor{gray}{23.10--25.17}&0.00, 0.01\\
512&1&\cellcolor{gray}{76.45--78.85}&0.10--0.14&\cellcolor{gray}{21.03--23.41}&0.00, 0.01\\
128&1&78.63--81.79 &0.08--0.12&18.12--21.26&0.00, 0.01\\\hline
24576&2&\cellcolor{gray}{54.03--54.43}&\cellcolor{gray}{1.35--1.52}&\cellcolor{gray}{43.98--44.41}&\cellcolor{gray}{0.11, 0.15}\\
1024&2&\cellcolor{gray}{74.47--76.18} &\cellcolor{gray}{0.15--0.19}&\cellcolor{gray}{23.66--25.35}&0.00, 0.01\\
512&2&\cellcolor{gray}{76.35--78.75}&0.11--0.16&\cellcolor{gray}{21.11--23.50}&0.00, 0.01\\
128&2&78.67--81.73 &0.09--0.13&18.15--21.21&0.00, 0.01\\\hline
24576&3&\cellcolor{gray}{51.60--51.96}&\cellcolor{gray}{2.21--2.39}&\cellcolor{gray}{45.52--45.90}&\cellcolor{gray}{0.19, 0.24}\\
1024&3&\cellcolor{gray}{74.24--76.20}&\cellcolor{gray}{0.18--0.23}&\cellcolor{gray}{23.59--25.53}&0.00, 0.02\\
512&3&\cellcolor{gray}{76.33--78.59}&\cellcolor{gray}{0.13--0.18}&\cellcolor{gray}{21.25--23.52}&0.00, 0.01\\
128&3&78.63--82.03 &0.09--0.13&17.88--21.24&0.00, 0.01\\\hline

\end{tabular}
\caption{Random and Burst Fault in $\kappa$ for CKKS}
\label{tab:ckks:q}
\end{table}

\begin{table}[!htb]
\centering
\begin{tabular}{|c|c|c|c|c|c|}
\hline
$\lambda$ & $\phi$ & $\chi_0$ (\%) & $\chi_1$ (\%) & $\chi_2$ (\%) & $\chi_3$ (\%) \\ \hline
$\times$&0&\textbf{80.06--83.81} &\textbf{0.07--0.11}&\textbf{16.11--19.85}&\textbf{0.00, 0.01}\\\hline
\multicolumn{6}{|c|}{\textbf{Random Fault in r}} \\ \hline
24576&1&\cellcolor{gray}{59.08--59.45}&\cellcolor{gray}{0.63--0.74}&\cellcolor{gray}{39.80--40.18}&\cellcolor{gray}{0.04, 0.07}\\
1024&1&\cellcolor{gray}{74.72--76.46}&\cellcolor{gray}{0.12--0.16}&\cellcolor{gray}{23.41--25.14}&0.00, 0.01\\
512&1&\cellcolor{gray}{76.44--78.82} &0.10--0.14&\cellcolor{gray}{21.07--23.44}&0.00, 0.01\\
128&1&78.69--81.81 &0.08--0.12&18.11--21.21&0.00, 0.01\\\hline
24576&2&\cellcolor{gray}{54.03--54.42}&\cellcolor{gray}{1.37--1.50}&\cellcolor{gray}{44.01--44.41}&\cellcolor{gray}{0.11, 0.16}\\
1024&2&\cellcolor{gray}{74.47--76.32}&\cellcolor{gray}{0.14--0.19}&\cellcolor{gray}{23.51--25.36}&0.00, 0.01\\
512&2&\cellcolor{gray}{76.33--78.61}&\cellcolor[gray]{0.8}{0.11--0.16}&\cellcolor{gray}{21.27--23.53}&0.00, 0.01\\
128&2&78.66--81.82 &0.09--0.13&18.08--21.21&0.00, 0.01\\\hline
24576&3&\cellcolor{gray}{51.59--51.98}&\cellcolor{gray}{2.22--2.39}&\cellcolor{gray}{45.49--45.89}&\cellcolor{gray}{0.19, 0.24}\\
1024&3&\cellcolor{gray}{74.23--76.04}&\cellcolor{gray}{0.18--0.23}&\cellcolor{gray}{23.74--25.55}&0.01, 0.02\\
512&3&\cellcolor{gray}{76.28--78.59} &\cellcolor{gray}{0.13--0.18}&\cellcolor{gray}{21.26--23.55}&0.00, 0.01\\
128&3&78.69--81.74 &0.09--0.13&18.16--21.19&0.00, 0.01\\\hline
\multicolumn{6}{|c|}{\textbf{Burst Fault in r}} \\ \hline
24576&1&\cellcolor{gray}{59.09--59.47}&\cellcolor{gray}{0.64--0.74}&\cellcolor{gray}{39.79--40.16}&\cellcolor{gray}{0.04, 0.08}\\
1024&1&\cellcolor{gray}{74.56--76.49} &\cellcolor{gray}{0.12--0.16}&\cellcolor{gray}{23.37--25.30}&0.00, 0.01\\
512&1&\cellcolor{gray}{76.35--78.75}&0.10--0.14&\cellcolor{gray}{21.12--23.52}&0.00, 0.01\\
128&1&78.78--82.03 &0.08--0.12&17.87--21.10&0.00, 0.01\\\hline
24576&2&\cellcolor{gray}{54.03--54.42}&\cellcolor{gray}{1.37--1.51}&\cellcolor{gray}{44.00--44.39}&\cellcolor{gray}{0.11, 0.15}\\
1024&2&\cellcolor{gray}{74.47--76.28}&\cellcolor{gray}{0.14--0.20}&\cellcolor{gray}{23.56--25.35}&0.00, 0.01\\
512&2&\cellcolor{gray}{76.31--78.72} &\cellcolor[gray]{0.8}{0.11--0.16}&\cellcolor{gray}{21.15--23.55}&0.00, 0.01\\
128&2&78.53--81.93 &0.08--0.13&17.97--21.35&0.00, 0.01\\\hline
24576&3&\cellcolor{gray}{51.60--51.98}&\cellcolor{gray}{2.21--2.39}&\cellcolor{gray}{45.54--45.87}&\cellcolor{gray}{0.19, 0.24}\\
1024&3&\cellcolor{gray}{74.22--76.21}&\cellcolor{gray}{0.18--0.24}&\cellcolor{gray}{23.58--25.55}&0.01, 0.02\\
512&3&\cellcolor{gray}{76.23--78.66}&\cellcolor{gray}{0.13--0.18}&\cellcolor{gray}{21.18--23.59}&0.00, 0.01\\
128&3&78.61--81.66 &0.09--0.13&18.24--21.28&0.00, 0.01\\\hline
\multicolumn{6}{|c|}{$\lambda:$ \# Faulty Loops, $\phi:$ \# Faulty bits, l=32, w=8, q=1811939329, n=4096} \\ \hline
\end{tabular}
\caption{Random and Burst Fault in r for CKKS}
\label{tab:ckks:r}
\end{table}

\section{Result \& Discussion}
\label{sec:result}
The proposed protected word-wise $BMM$ is implemented in VHDL on an Artix-7 FPGA using Vivado 23.2. The proposed wordwise $BMM$ is implemented with kyber and CKKS parameter set. The results of the proposed methods are compared in terms of overhead and error-detection efficiency in the following two sections, namely Sec. \ref{sec:overhead} and Sec. \ref{sec:error:efficiency}, respectively.
\subsection{Overhead}
\label{sec:overhead}
The overhead of the proposed fault-detection scheme for the $BMM$ is evaluated using three approaches: (i) overhead analysis across different $BMM$ word sizes for the Kyber and CKKS parameter sets, and (ii) overhead comparison of the proposed $BMM$ against existing fault-detection solutions reported in the literature and (iii) comparison of the overhead of the proposed protected $BMM$ placed inside the $NTT$ with existing protected $NTT$ schemes
\subsubsection{Impact of Wordsize and parameter set on Hardware Overhead}
As shown in Table \ref{tab:bmm_srm_results}, the Kyber $BMM$ with word sizes 4 and 6 incurs resource overheads of 18.9\% and 20.5\%, respectively. However, the full Kyber $BMM$ implementation without word-wise partitioning incurs only an 8.3\% resource overhead. Similarly, for CKKS, the proposed wordwise $BMM$ with word sizes 8 and 16 incurs resource overheads of 10.9\% and 7.9\%, respectively. However, the full CKKS $BMM$ implementation without word-wise partitioning incurs only a 6.13\% resource overhead. The energy overhead of the proposed method is approximately 1\% for all Kyber and CKKS variants.

\begin{table*}[!htb]
\centering
\begin{tabular}{|p{2.1cm}|p{2.3cm}|p{1.3cm}|c|c|c|p{1.4cm}|c|p{0.7cm}|p{0.7cm}|}
\hline
\textbf{Application} & \textbf{Config (l, w, q)} & \textbf{SEC} & \textbf{Slice} & \textbf{LUT/FF/DSP} & \textbf{Power (mW)} & \textbf{Energy (nJ)} & \textbf{CC} & \textbf{CP (ns)} & \textbf{CLK (MHz)} \\ \hline

BMM & Kyber, 12, 4, 3329 & 238  & 38 & 100 / 47 / 2 & 106 &  10.6& 10 & 9.46 & 100 \\ \hline
BMM+SRM & Kyber, 12, 4, 3329 &283 (\textbf{18.9\%$\uparrow$})  & 83 & 163 / 156 / 2 & 107  & 10.7 (\textbf{0.9\%$\uparrow$}) & 10 & 9.33 & 100 \\ \hline

BMM & Kyber, 12, 6, 3329 & 243 & 43 & 113 / 47 / 2 & 108 & 5.4 & 5 & 9.32 & 100 \\ \hline
BMM+SRM & Kyber, 12, 6, 3329 & 293 (\textbf{20.5\%$\uparrow$})  & 93 & 176 / 146 / 2 & 109 & 5.45 (\textbf{0.27\%$\uparrow$})& 5 & 9.42 & 100 \\ \hline

BMM & Kyber, 12, 12, 3329 & 513 & 13 & 31 / 26 / 5 & 106 & 2.54 & 2 & 11.67 & 83 \\ \hline
BMM+SRM & Kyber, 12, 12, 3329 & 556 (\textbf{8.3\%$\uparrow$}) & 56 & 94 / 112 / 5 & 107 & 2.56 (\textbf{0.7\%$\uparrow$}) & 2 & 11.53 & 83 \\ \hline

BMM & CKKS, 32, 8, 1811939329 & 374 & 74 & 234 / 84 / 3 & 130 & 22.1 & 17 & 8.78 & 100 \\ \hline
BMM+SRM & CKKS, 32, 8, 1811939329 &  415 (\textbf{10.9\%$\uparrow$})& 115 & 295 / 183 / 3 & 131 & 22.27 (\textbf{0.7\%$\uparrow$}) & 17 & 8.87 & 100 \\ \hline

BMM & CKKS, 32, 16, 1811939329 & 476 & 76 & 184 / 44 / 4 & 131 & 6.55 & 5 & 8.71 & 100 \\ \hline
BMM+SRM & CKKS, 32, 16, 1811939329 & 514 (\textbf{7.9\%$\uparrow$}) & 114 & 245 / 143 / 4 & 131.5 & 6.75 (\textbf{3\%$\uparrow$})  & 5 & 9.04 & 100 \\ \hline

BMM & CKKS, 32, 32, 1811939329 & 652 & 52 & 163 / 66 / 6 & 131 & 2.62 & 2 & 8.36 & 100 \\ \hline
BMM+SRM & CKKS, 32, 32, 1811939329 & 692 (\textbf{6.13\%$\uparrow$}) & 92 & 222 / 165 / 6 & 132 & 2.64 (\textbf{0.7\%$\uparrow$})& 2 & 9.12 & 100 \\ \hline
Aghapour~et~al. \cite{aghapour_barrett}~REMO+BMM & FHE, 1024, 32, - & 5146 (\textbf{17.1\%$\uparrow$}) & - &13720 / 10522 / 4 & 43 & 17.91 (\textbf{21.87\%$\uparrow$}) & 7365 & 23.2 & 40 \\ \hline
Baidya~et~al.\cite{baidya} RESWO+BMM & PQC, 12, 4, 3329 & 456 (\textbf{85.3\%$\uparrow$}) & - & 162/ 70 / 4 & - & 10.8  & 10 & 9.63 & 100 \\ \hline
Baidya~et~al.\cite{baidya} RESO+BMM & PQC, 12, 4, 3329 & 463 (\textbf{88.2\%$\uparrow$}) & - & 147/ 75 / 4 & - & 10.9  & 10 & 9.64 & 100 \\ \hline
Baidya~et~al.\cite{baidya} RENO+BMM & PQC, 12, 4, 3329 & 462 (\textbf{87.8\%$\uparrow$}) & - & 174/ 72 / 4 & - & 10.8  & 10 & 9.8 & 100 \\ \hline
Aghapour et al. \cite{aghapour_barrett} BMR & FHE, 2048, 32, - & 6556 (\textbf{19.4\%$\uparrow$}) & - &18801 / 11646 / 4 & & 3.04 (\textbf{31.73\%$\uparrow$}) & 2264 & 17.9 & 50 \\ \hline
\end{tabular}
\caption{Hardware Overhead Comparison of BMM and BMM+SRM}
\label{tab:bmm_srm_results}
\end{table*}

 \begin{table*}[!htbp]
 		\centering
\begin{tabular}{|m{1.2cm}|>{\centering\arraybackslash}m{2.6cm}|
>{\centering\arraybackslash}m{2.4cm}|
>{\centering\arraybackslash}m{1.2cm}|
>{\centering\arraybackslash}m{1.2cm}|
>{\centering\arraybackslash}m{1.2cm}|
>{\centering\arraybackslash}m{0.7cm}|
>{\centering\arraybackslash}m{1.2cm}|}
	\hline
	\textbf{Work}& \textbf{Type of Fault } & \textbf{Platform}& \multicolumn{3}{c|}{\textbf{Overhead(\%) }}&  \textbf{Frq.}& \textbf{(\%) Error} \\ \cline{4-6}
	
	& \textbf{Detection \& Target HW} &  &\textbf{Area} & \textbf{Delay} & \textbf{Energy} &  (MHz.) &\textbf{Coverage} \\ \hline  
   Sarker et al.\cite{sarker}  &RENO v1/ v2/ v3 for Falcon and NTRU $NTT$ & Spartan7 FPGA& 20.2\%$\uparrow$*/ 15.3\%$\uparrow$*/ 21.5\%$\uparrow$*  & 8.46\%$\uparrow$/ 15.88\%$\uparrow$)/ 13.71\%$\uparrow$   & 15.6\%$\uparrow$/ 7.6\%$\uparrow$/ 11.2\%$\uparrow$   &50 &99.51/99.67 /99.41 \\ \hline
 Sarker et al. \cite{sarker}  &RENO v1/v2/v3 for Falcon and NTRU $NTT$& Zynq FPGA& 24.24\%$\uparrow$*/ 14\%$\uparrow$*/ 17.8\%$\uparrow$*  & 9.32\%$\uparrow$/ 19.66\%$\uparrow$/ 21.78\%$\uparrow$  & 20.47\%$\uparrow$/ 13.27\%$\uparrow$/ 17.26\%$\uparrow$  &  50 & 99.51/99.67 /99.41 \\ \hline
  Ahmadi et al. \cite{ahmadi} & Correlated Check sum for Kyber $NTT$ &Artix-7 FPGA& 25.1\%$\uparrow$  &  41.3\% $\uparrow$&51.5\%$\uparrow$ & 140 & \textbf{$\sim 53$ -- $\sim 99.99$} \\ \hline
    Ravi et al. \cite{ravi_ntt} & Checking Entropy &ARM Cortex M4& NR  &  NR&NR &NR &NR \\ \hline
   Bauer et al. \cite{sven} &Polynomial Evaluation and Interpolation for Dilithium $NTT$&ARM Cortex M4& NR  & 72 &NR &NR &NR \\ \hline
    Jati et al. \cite{jati} & Randomized Memory Address for Kyber $NTT$ &Artix 7 FPGA& NR  &  NR&NR& NR  &NR \\ \hline
    Baidya et al. \cite{baidya} & RESWO for Kyber $NTT$ (wordwise) &Artix 7 FPGA& 24.29*  &  3.8&1.45 &100 &99.97 \\ \hline
   Baidya et al. \cite{baidya} & RESO for Kyber $NTT$ (wordwise) &Artix 7 FPGA& 34.34  &  5.8&2.97 &100 &99.97 \\\hline
    Baidya et al. \cite{baidya} & RENO for Kyber $NTT$ (wordwise) &Artix 7 FPGA& 24.71*  &  0.75&2.14 &100 &99.97 \\ \hline

    \textbf{Our} & SRM for Kyber $NTT$ (wordwise) &Artix 7 FPGA& \textbf{5.2}*  &  \textbf{1.2}&\textbf{<1} &100 &$33$ -- $100$ \\ \hline
    \textbf{Our} & SRM for CKKS $NTT$ (wordwise) &Artix 7 FPGA& \textbf{1.2}*  &  \textbf{1.3} &\textbf{<1} &100 &$36$ -- $100$ \\ \hline

\end{tabular}
\begin{tablenotes}
\item  $*$ If slice usage is not available we calculate SEC=0.25 $\times$ LUTs + 0.125 $\times$ FFs + 100 $\times$ DSPs + 200 $\times$ BRAMs \cite{sec}.~~ NR = Not Reported
\end{tablenotes}
\vspace{2pt}
 		\caption{Overhead Comparison with $NTT$ literature}
 	\label{tab:lit}
 \end{table*}

\subsubsection{Comparison with Existing Fault-Detection Schemes for BMM}
To the best of our knowledge, only two existing works, \cite{baidya} and \cite{aghapour_barrett}, have proposed fault-detection mechanisms for BMM based on recomputation. Compared with these approaches, the proposed SRM-based fault-detection scheme achieves significantly lower hardware and energy overhead while maintaining effective fault-detection capability. Recomputation based approaches RESWO, RESO, and RENO proposed by Baidya et al. \cite{baidya} incur approximately 85\%–88\% slice overhead, the proposed SRM requires only 18.9\% additional slices for the equivalent Kyber configuration (l=12,w=4,q=3329). Similarly Barrett Modular Redundancy (BMR) scheme of Aghapour et al. has area overheads of 17.1\%–19.4\% and energy overheads of 21.87\%–31.73\%, the proposed $SRM$ achieves comparable or lower area overhead while reducing the energy overhead to below 3\%. This low overhead of our proposed scheme is achieved through its statistical monitoring methodology, unlike existing approaches that rely on arithmetic recomputation or redundant hardware units. The statistical observation on the intrinsic reduction-path execution of the word-wise $BMM$ by our proposes $SRM$ scheme achieves minimal additional logic, leading to substantially lower area and energy consumption while preserving the original performance characteristics of the $BMM$ architecture.

\subsubsection{Overhead of BMM Embedded in $NTT$}
The proposed protected $BMM$ with $SRM$ is embedded in the Kyber and CKKS $NTT$s, resulting in approximately 5.2\% and 1.2\% area overhead, 1.2\% and 1.3\% delay overhead, and <1\% and <1\% power overhead, respectively. We identified six additional studies \cite{sarker,ahmadi,ravi,sven,jati,baidya} in which the authors proposed $NTT$ protection schemes targeting different subcomponents of the $NTT$. Sarker et al. protected modular multiplication, modular addition, and modular subtraction of $NTT$ using a recomputation-based fault-detection technique, RENO on Spartan 7 and Zynq FPGA. Ahmadi et al. \cite{ahmadi} implemented a correlated checksum-based fault-detection scheme for $NTT$ polynomials on an Artix-7 FPGA. Similar increments and decrements in checksum-based fault detection may not be able to detect fault occurrences. Baidya et al. \cite{baidya} implemented three recomputation-based fault-detection techniques, namely RESO, RENO, and RESWO, for the protection of the $BMM$ modular multiplication unit employed within the $NTT$. Jati et al. \cite{jati} implemented a randomized memory addressing scheme for the Kyber $NTT$. This randomization is feasible because the order in which twiddle factors are read from memory within a stage does not affect the correctness of the $NTT$ computation. The resulting randomized memory access pattern mitigates soft analytical side-channel attacks (SASCA). Bauer et al. \cite{sven} employed a polynomial evaluation and interpolation technique for fault detection at the $NTT$ output, resulting in a substantial timing overhead of 72\% on the ARM Cortex-M4 platform. Ravi et al. \cite{ravi_ntt} demonstrated a zeroizing attack on the $NTT$ running on the ARM Cortex-M4 platform, which can reduce the entropy of the secret polynomial used in lattice-based PQC. A detailed discussion is available in Sec.~\ref{sec:threat}. As a countermeasure, the authors proposed measuring the entropy of the $NTT$ output. As shown in Table \ref{tab:lit}, the proposed $SRM$ method incurs the lowest area, delay, and energy overhead among all existing protected $NTT$ schemes.
\subsection{Error Detection Efficiency}
\label{sec:error:efficiency}
The error detection efficiency is evaluated using two approaches: simulation-based fault injection and emulated fault injection.
\subsubsection{Error Simulation}
The error injection procedure was simulated using Python on Ubuntu 24.04 running on an Intel i5 processor, where each $NTT$ operation was executed 100,000 times for random and burst faults. For Kyber, a complete $NTT$ computation requires 1,024 $NTT$ iterations, resulting in 1,024 $BMM$ calls. Similarly, for CKKS, a complete $NTT$ computation requires 24,576 $NTT$ iterations, resulting in 24,576 $BMM$ calls. Since each $NTT$ iteration performs exactly one $BMM$ operation, the total number of $BMM$ calls is equal to the total number of $NTT$ iterations. For both Kyber and CKKS KeyGen operations, we varied the number of faulty $NTT$ iterations ($\lambda$) and the number of faulty bits ($\phi$), and investigated whether the metrics $\chi_0$, $\chi_1$, and $\chi_2$ deviated from their corresponding values under fault-free conditions. As reported in Table~\ref{tab:kyber:error} and Table~\ref{tab:ckks:error}, for Kyber and CKKS key generation, respectively, the proposed scheme achieves a 100\% fault detection rate for both random and burst fault injections under permanent fault conditions, irrespective of the number of faulty bits ($\phi$). Here, a permanent fault refers to a fault that persists in any of the intermediate registers $c$, $\kappa$, or $r$ throughout all $NTT$ iterations. For transient faults, the number of affected $NTT$ iterations is reduced from the total 1,024 iterations in Kyber $NTT$ and 24,576 iterations in CKKS $NTT$. If a fault persists for at least 128 out of the 1,024 Kyber $NTT$ iterations, the proposed fault model still achieves a 100\% fault detection rate. Similarly, for CKKS $NTT$, if the fault persists for at least 512 out of the 24,576 $NTT$ iterations, the proposed model is still able to detect the fault with 100\% accuracy. When the number of $NTT$ iterations affected by a transient fault falls below 64 for Kyber and below 128 for CKKS, the fault detection rate of the proposed scheme begins to decrease. 
\subsubsection{Error Emulation}
In the fault emulation process, a fault injector hardware is built which directly connected with the three registers $c$, $\kappa$ and $r$ where we consider fault locations. This fault injector may be introduced in the design in the form of hardware Trojan. Under faulty conditions, the fault injector can modify any targeted bit of $c$, $\kappa$, or $r$ in either random-fault or burst-fault mode. The simulation results are reported in Tables~\ref{tab:kyber:c}, Tables~\ref{tab:kyber:q}, Tables~\ref{tab:kyber:r}, Tables~\ref{tab:ckks:c}, Tables~\ref{tab:ckks:q} and Tables~\ref{tab:ckks:r} as simulation enables the efficient evaluation of a significantly larger number of fault samples than hardware emulation.

\begin{table}[!htb]
\centering
\begin{tabular}{|p{1.0cm}|c|c|p{1.2cm}|p{1.9cm}|}
\hline
 \textbf{Fault Type} & $\lambda$ & $\phi$ & \textbf{Fault Location} & \textbf{Efficiency} (\%) \\ \hline

    & 128--1024 & 1--11 & c & 100 \\ \cline{2-5}
    & 64    & 1-11  & c & $\sim$96.12--100     \\ \cline{2-5}
    & 32    & 1  & c & $\sim$30.62     \\ \cline{2-5}
    & 32    & 2--11  & c & $\sim$60.77--100     \\ \cline{2-5}

Random & 128--1024 & 1--11 & $\kappa$ & 100 \\\cline{2-5}
  & 64    & 1--11  & $\kappa$ & $\sim$ 36    \\ 
 \cline{2-5}
  & 128--1024 & 1--11 & r & 100 \\ \cline{2-5}
        & 64    & 1  & r & $\sim$46.12   \\ \cline{2-5}
       & 64    & 2--11  & r & $\sim$90.39--100    
       \\ \hline
     & 128--1024 & 1--11 & c & 100 \\ \cline{2-5}
        & 64    & 1--11  & c & $\sim$96.49--100     \\ \cline{2-5}
        & 32    & 1  & c & $\sim$29.40     \\ \cline{2-5}
       & 32    & 2--11  & c & $\sim$29.87--79.21     \\ \cline{2-5}
  Burst & 128--1024 & 1--11 & $\kappa$ & 100 \\ \cline{2-5}
       & 64    & 1  & $\kappa$ & $\sim$33.18    \\ \cline{2-5}
      & 64    & 2--11  & $\kappa$ & $\sim$33--66      \\ \cline{2-5}
      & 128--1024 & 1--11 & r & 100 \\ \cline{2-5}
      & 64    & 1  & r & $\sim$46.68   \\ \cline{2-5}
      & 64    & 2--11  & r & $\sim$67.67--100     \\ \hline
\end{tabular}
\caption{BMM Fault Coverage for Kyber (l=12, w=4, q=3329)}
\label{tab:kyber:error}
\end{table}

\begin{table}[!htb]
\centering
\begin{tabular}{|p{0.8cm}|c|c|p{1.2cm}|p{1.9cm}|}
\hline
 \textbf{Fault Type} & $\lambda$ & $\phi$ & \textbf{Fault Location} & \textbf{Efficiency} (\%) \\ \hline

    & 512--24576 & 1--11 & c & 100 \\ \cline{2-5}
    & 128    & 1-11  & c & $\sim$50.89--    \\ \cline{2-5}

Random & 128--1024 & 1--11 & $\kappa$ & 100 \\\cline{2-5}
  & 64    & 1--11  & $\kappa$ & $\sim$ 36    \\ 
 \cline{2-5}
  & 128--1024 & 1--11 & r & 100 \\ \cline{2-5}
        & 64    & 1  & r & $\sim$46.12   \\ \cline{2-5}
       & 64    & 2--11  & r & $\sim$90.39--100    
       \\ \hline
     &  512--24576 & 1--11 & c & 100 \\ \cline{2-5}
        & 128    & 1--11  & c & $\sim$50.78--100     \\ \cline{2-5} 
  Burst & 512--24576 & 1--11 & $\kappa$ & 100 \\ \cline{2-5}
       & 128    & 1-11  & $\kappa$ & $\sim$50.89--100    \\ \cline{2-5}
      & 512--24576 & 1--11 & r & 100 \\ \cline{2-5}
      & 128    & 1--11  & r & $\sim$52.23--100  \\
       \hline
\end{tabular}
\caption{BMM Fault Coverage for CKKS (l=32, w=8, q=1811939329)}
\label{tab:ckks:error}
\end{table}

\section{Conclusion}
$BMM$ is the most critical hardware component of $NTT$ used in lattice based cryptogarphy. Since the $BMM$ consumes a significant portion of the hardware resources and energy budget of lattice-based PQC and FHE accelerators, it becomes an attractive target for attackers. This paper proposes and implements a novel statistical fault detection method based on the reduction path execution patterns of a word-wise BMM. To the best of our knowledge, the proposed scheme achieves comparable fault-detection capability with lower hardware-resource, delay and energy overhead than existing BMM specific fault-detection techniques reported in the literature. The proposed schemes can detect 100\% of permanent faults and are also effective in detecting transient faults. Furthermore, when the proposed $SRM$ and $BMM$ are incorporated into the $NTT$ architecture, the resulting design exhibits the lowest area, energy, and timing overhead among all existing architectures considered in this work, to the best of our knowledge.\\
\textbf{To promote open research and reproducibility, implementation and associated test
cases of wordwise $BMM$ are uploaded on a public GitHub repository \footnote{https://github.com/rourabpaul1986/wordwise\_pipelined\_barrett}, enabling independent verification and thorough validation of the proposed methodology.}
\label{sec:con}

\bibliographystyle{unsrt}  
\bibliography{IEEEexample}

\end{document}